%% file: Thesis.tex
\documentclass[en,founder]{simplepaper}
\begin{document}
\title{Complex Principle Kurtosis Analysis}
\author{Liangliang Zhu\footnote{School of Computer Science and Information Engineering, Hefei University of Technology, Hefei 230601, Anhui Province, China}, Zhebin Song\footnotemark[1], Xuesen Zhang\footnote{Key Laboratory of aperture array and space exploration, No.38 Research Institute of CETC, Hefei 230031, Anhui Province, China}, and Meibin Qi\footnotemark[1] \footnote{Corresponding author: qimeibin@163.com}}
\date{\today}
\maketitle
\input{tex/1-abstract}
\input{tex/2-introduction}
\input{tex/3-background}

\input{tex/4-method}
\input{tex/5-experiments.tex}
\input{tex/6-conclusion.tex}
\printbibliography[heading=simplepaper]
\end{document}

%% file: tex/1-abstract.tex
\begin{abstract}
    Independent component analysis (ICA) is a fundamental problem in the field of signal processing, and numerous algorithms have been developed to address this issue. The core principle of these algorithms is to find a transformation matrix that maximizes the non-Gaussianity of the separated signals. Most algorithms typically assume that the source signals are mutually independent (orthogonal to each other), thereby imposing an orthogonal constraint on the transformation matrix. However, this assumption is not always valid in practical scenarios, where the orthogonal constraint can lead to inaccurate results. Recently, tensor-based algorithms have attracted much attention due to their ability to reduce computational complexity and enhance separation performance. In these algorithms, ICA is reformulated as an eigenpair problem of a statistical tensor. Importantly, the eigenpairs of a tensor are not inherently orthogonal, making tensor-based algorithms more suitable for nonorthogonal cases. Despite this advantage, finding exact solutions to the tensor's eigenpair problem remains a challenging task. In this paper, we introduce a non-zero volume constraint and a Riemannian gradient-based algorithm to solve the tensor's eigenpair problem. The proposed algorithm can find exact solutions under nonorthogonal conditions, making it more effective for separating nonorthogonal sources. Additionally, existing tensor-based algorithms typically rely on third-order statistics and are limited to real-valued data. To overcome this limitation, we extend tensor-based algorithms to the complex domain by constructing a fourth-order statistical tensor. Experiments conducted on both synthetic and real-world datasets demonstrate the effectiveness of the proposed algorithm.
\end{abstract}

%% file: tex/2-introduction.tex
\section{Introduction}

Blind Source Separation (BSS), which aims to isolate source components from mixed signals, is widely used in various fields, including audio signal processing \cite{liutkus2013overview, makino2005blind, makino2006blind}, radar signal processing \cite{fabrizio2014blind, guo2017time, cardoso1993blind}, image processing \cite{meganem2014linear, golbabaee2010distributed}, and so on \cite{hirayama2015unifying, choi2005blind}. More specifically, in audio signal processing, the most common task of BSS is to extract source signals without specific prior information about sources \cite{vincent2006performance}. In radar signal processing, BSS can be used to perform beamforming without knowing array manifold \cite{cardoso1993blind} or separate the interference signals without knowing any prior knowledge of interference \cite{huang2004radar}.

Since the BSS problem was first introduced by C.Jutten in his PhD thesis and reproduced in a concise form in \cite{jutten1988solution}, numerous algorithms has been developed. Moreover, the application scenarios of these algorithms can be roughly divided into two categories, namely instantaneous mixture and convolutive mixture \cite{murata1998line, makino2005blind}. In the instantaneous mixture case, observed signals are the linear combinations of source signals, hence, can be denoted by the multiplication of the mixture matrix and the signal matrix. Meanwhile, convolutive mixture indicates that signals are mixed in a convolutive manner, i.e., each observe signal can be modeled as a weighted sum of different time delayed source signals \cite{makino2005blind, kemiha2017complex}. A typical instance of the convolutive mixture is the occurrence of reverberation. Notice that convolutive mixing can be transformed into instantaneous mixing by applying specific transformations to the observed signals, such as Fast Fourier Transform (FFT) and Short-Time Fourier Transform (STFT) \cite{kurita2000evaluation, hiroe2006solution, kim2006independent, yatabe2021determined}. Therefore, in this manuscript, we mainly focus on the instantaneous mixing case.

Independent component analysis (ICA), as a special case of BSS, is able to deal with non-Gaussian signals, hence, becomes one of the most popular reasearch topics in recents years. Initially, scholars use neural networks with stochastic gradient descend (SGD) to perform ICA. However, these methods are critically flawed due to their slow convergence and the difficulty of selecting appropriate learning rate parameters \cite{hyvarinen1997fast, bell1995information, oja1995signal}. In \cite{hyvarinen1997fast}, the author introduced the fast fixed-point algorithm for independent component analysis (FastICA), which later became one of the most widely used algorithms. Generally, FastICA tries to find a linear unmixing vector that maximizing or minimizing some particular non-Gaussianity metric of the umixed signals. FastICA adopts an iterative algorithm that can be considered as an approximation of the Newton's method to find the optimal unmixing vectors one by one. Consequently, it does not require any learning rate parameter and achieves significantly faster convergence. Notice that in ICA, the source signals are assumed to be statistical independent, hence in FastICA, the unmixing matrix is constrained to be orthogonal. Moreover, in FastICA, the metric to measure the non-Gaussianity of the unmixed signals can be variously selected, e.g., skewness, kurtosis, negentropy, and so on. Unfortunately, the original FastICA algorithm can only process real signals, it is limited in handling with data with complex values such as the radar signal. To expand the scope of applications, Bingham proposed the complex version of FastICA, namely cFastICA, \cite{Bingham2000A}. Similar to FastICA, cFastICA also adopts the orthogonal constraint to prevent the unmixing vectors from converging into the same point. Meanwhile, in \cite{cardoso1993blind}, the author noticed that when using the fourth-order cumulants as the metric of non-Gaussianity, the problem can be solved by jointly diagonalizing a set of matrices. Such a diagonalization task can be efficiently approached by the famous singular value decomposition (SVD) algorithm. Thus, the proposed algorithm named Joint Approximate Diagonalization of Eigen-matrices (JADE) is usually applied in scenarios that demand high real-time performance, such as audio signal and radar signal processing. Furthermore, because of the usage of SVD, the obtained unmixing matrix in JADE is also orthogonal.

Although ICA algorithms are varied, their cost function can be divided into two categories, one is based on information entropy and the other is based on high-order statistics. In \cite{geng2014principal}, the authors found that the ICA problem can also be solved by finding the eigenpairs of the statistical tensor of the mixed signals, and proposed a new algorithm named Principal Skewness Analysis (PSA), where the eigenvectors are exactly the desired unmixing vectors. Although the iterative formula of PSA was proved to be equivalent to the iterative formula of FastICA based on 3-order statistics (skewness), the introduction of the statistical tensor demonstrates the potential for of an algorithm with lower computational complexity and improved performance. In \cite{geng2020npsa}, the author indicated that the eigenvectors of a high order statistical tensor were naturally nonorthogonal. However, most existing algorithms, such as FastICA and JADE, constraint the unmixing vectors to be orthogonal to each other. Therefore, existing algorithms cannot get the accurate solution in most cases. To overcome this problem, the authors of PSA further proposed a nonorothgonal constraint based on Kronnecker product. The algorithm, named Nonorthogonal Principal Skewness Analysis (NPSA), is able to obtain nonorothgonal eigenvectors, thus achieves a better separation performance than PSA. However, the obtained eigenpairs in NPSA are still not the accurate ones. Although algorithms based on information entropy \cite{yatabe2021determined, kim2006independent} naturally contain nonorthogonal constraints, the discussion of entropy estimation is only limited in independent signals and entropy for dependent source signals are hard to measure, which results in the inaccurate solution in nonorthogonal cases.

To expand the application scenarios, such as radar signal processing, of tensor-based BSS algorithm, we study the high-order complex extension of PSA and NPSA. Besides, neither PSA nor NPSA discussed the impact of noise power on solution. Therefore, we derive the influence of the noise in detail in this paper. Based on the two findings above, we aim to propose a high-performance and widely applicable BSS algorithm.

%% file: tex/3-background.tex
\section{Background}
In this section, we will give a brief review of the formulations and deficiencies of FastICA, PSA and JADE. Assume that there are \( N \) source signals \( \mathbf{S} = \begin{bmatrix} \bm{s}_1 & \bm{s}_2 & \cdots & \bm{s}_L \end{bmatrix} \in \mathbb{R}^{N \times L}\), where \( L \) is number of samples. If the instantaneous mixture model is considered, then the recorded signals can be expressed as a linear combination of the source signals, i.e.,
\[
    \mathbf{X} = \mathbf{A}\mathbf{S} = \begin{bmatrix}
        \mathbf{x}_1 & \mathbf{x}_2 & ... & \mathbf{x}_{L}
    \end{bmatrix}\in \mathbb{R}^{M \times L},
\]
where \( \mathbf{A} \in \mathbb{R}^{M \times N}\) is the mixing matrix and \( M \) is the number of recorded signals. Generally, the mixing matrix is unknown, and \( M \geq N \). For convenience, we let \( M = N \) in this paper. The goal of the mentioned ICA algorithms is to estimate the unmixing matrix \( \mathbf{W} \in  N \times N  \) such that the source signals can be recovered by \( \mathbf{S} = \mathbf{W}^{\mathrm{T}} \mathbf{X} \). The details are given in the following subsections.

% In this paper, the matrix and the vector is denoted by bold regular upper font and bold regular lower font respectively, e.g., $\mathbf{W}$ and $\mathbf{x}$, the scalar is denoted by italic font, e.g., $N$ and the set is denoted by regular font, e.g., $\mathrm{R}$. The element of $\mathbf{A} \in \mathrm{R}^{I \times J}$ is denoted as $a_{ij}$ where $i \in \{1, 2,  \cdots , I\}$ and $j \in \{1, 2,  \cdots , J\}$. Mixed signals are defined as $\mathbf{X} = \begin{bmatrix} \mathbf{x}_1 & \mathbf{x}_2 &  \cdots  & \mathbf{x}_L \end{bmatrix} \in \mathrm{R}^{N \times L}$, where $\mathbf{x}_i$ is a $N \times 1$ vector, $N$ is the number of mixed channels, $\mathrm{R}$ is the domain of definition and $L$ is the number of sampling points. Define mixing matrix as $\mathbf{A} \in R^{N \times M}$ and source signal matrix as $\mathbf{S} \in R^{M \times L}$ where $M$ is the number of source signals. Therefore, Mixed signals can be represented as $\mathbf{X} = \mathbf{A}\mathbf{S}$ in instantaneous mixing cases.

\subsection{FastICA}
The basic idea of FastICA is to maximize the non-gaussianity of the estimated source signals. Denote the non-gaussianity metric function as \( G(\cdot) \), which can be performed on each element of the input vector to obtain a new vector, then the optimization problem of FastICA can be formulated as
\[
    \left\{
    \begin{aligned}
        \max        & \quad G(\mathbf{w}^{\mathrm{T}} \mathbf{X}) \mathbf{1}_L \\
        \text{s.t.} & \quad \mathbf{w}^{\mathrm{T}} \mathbf{w} = 1
    \end{aligned}
    \right..
\]
where \( \mathbf{1}_L \) is a \( L \times 1 \) vector with all elements equal to 1, and \( \mathbf{w} \in \mathbb{R}^{N \times 1} \) is the unmixing vector for a specific source signal. The iterative formula for \( \mathbf{w} \) can be obtained as follows
\begin{equation}
    \label{eq:FastICAFormular}
    \mathbf{w} \leftarrow \frac{1}{L} G''(\mathbf{w}^{\mathrm{T}} \mathbf{X})\mathbf{1}_L - \mathbf{X}G'(\mathbf{w}^{\mathrm{T}} \mathbf{X})^{\mathrm{T}}, \quad \mathbf{w} \leftarrow \frac{\mathbf{w}}{||\mathbf{w}||},
\end{equation}
where \( G'(\cdot) \) and \( G''(\cdot) \) are the first and second order derivatives of \( G(\cdot) \) respectively. The iterative formula is performed until the convergence condition is satisfied. Assume that \( \mathbf{W}_k = \begin{bmatrix} \mathbf{w}_1 & \mathbf{w}_2 & \cdots & \mathbf{w}_k \end{bmatrix} \) is the matrix consisting the first \( k \) obtained unmixing vectors, then to ensure that the \( k+1 \) unmixing vector will not converge into the previous \( k \) unmixing vectors, the following orthogonal constraint is enforced during the iterative process
\[
    \mathbf{w}_{k+1} \leftarrow \mathbf{P}_{\mathbf{W}_k} \mathbf{w}_{k+1},
\]
where \( \mathbf{P}_{\mathbf{W}_k} = \mathbf{I} - \mathbf{W}_k \mathbf{W}_k^{\mathrm{T}} \) is the orthogonal complement projection matrix, $\mathbf{I}$ is the identity matrix.
Moreover, it can be found from \cref{eq:FastICAFormular} that every iteration involves the calculation of the whole data matrix \( \mathbf{X} \). This could lead to high computing complexity when the number of samples \( L \) is large. To reduce the computing complexity, Geng \cite{geng2014principal, geng2020npsa} proposed PSA, where each iteration only involves the calculation of the third-order coskewness tensor of the data. The details of PSA will be given in the next subsection.

\subsection{PSA}
PSA tries to maximize the skewness of the estimated source signals, and the skewness of a vector \( \mathbf{x} = \begin{bmatrix} x_1 & x_2 & \cdots & x_L \end{bmatrix}^{\mathrm{T}} \) is defined as
\[
    \text{skew}(\mathbf{x}) = \frac{1}{L} \sum_{i=1}^{L} \left( \frac{x_i - \bar{x}}{\sigma} \right)^3,
\]
where \( \bar{x} \) is the mean of \( \mathbf{x} \) and \( \sigma \) is the standard deviation of \( \mathbf{x} \). If the vector \( \mathbf{x} \) is prewhitened, i.e., \( \bar{x} = 0 \) and \( \sigma = 1 \), then the skewness of \( \mathbf{x} \) can be simplified as
\[
    \text{skew}(\mathbf{x}) = \frac{1}{L} \sum_{i=1}^{L} x_i^3.
\]

Assume that the data matrix \( \mathbf{X} \) is prewhitened, i.e., \( \mathbf{X} \mathbf{1}_L = \mathbf{0}_{N} \) and \( \frac{1}{L} \mathbf{X}^{\mathrm{T}} \mathbf{X} = \mathbf{I}_{N \times N} \). Then, after the projection of \( \mathbf{X} \) by the unmixing vector \( \mathbf{w} \), the skewness of the projected data can be expressed as
\begin{equation} \label{eq:skewness}
    \text{skew}(\mathbf{w}^{\mathrm{T}} \mathbf{X}) = \frac{1}{L} \sum_{i=1}^{L} (\mathbf{w}^{\mathrm{T}} \mathbf{x}_i)^3 ,
\end{equation}
By introducing the third-order coskewness tensor \( \mathcal{S} \in \mathbb{R}^{N \times N \times N}\), the skewness of the projected data can be rewritten as
\[
    \text{skew}(\mathbf{w}^{\mathrm{T}} \mathbf{X}) = \mathcal{S} \times_1 \mathbf{w} \times_2 \mathbf{w} \times_3 \mathbf{w}.
\]
Just like the covariance matrix, the third-order coskewness tensor \( \mathcal{S} \) contains the complete third-order statistical information of the data, and has the following definition
\[
    \mathcal{S}_{ijk} = \frac{1}{L} \sum_{l=1}^{L} X_{il} X_{jl} X_{kl},
\]
where \( X_{il} \) is the \( (i,l) \)th element of \( \mathbf{X} \). It should be noted that \( \times_n \) in \cref{eq:skewness} represents the $n$-mode product, which can be considered as the generalization of the matrix product. The $n$-mode product of a tensor $\mathcal{A} \in \mathbb{R}^{I_1 \times \cdots I_{n-1} \times I_n \times I_{n+1} \cdots \times I_N}$ with $\mathbf{U} \in \mathbb{R}^{I_n \times J}$ is denoted as $\mathcal{A} \times_n \mathbf{U} \in \mathbb{R}^{I_1 \times \cdots I_{n-1} \times J \times I_{n+1}\cdots \times I_N}$, and follows the calculation rule
\[
    (\mathcal{A} \times_n \mathbf{U})_{i_1 \cdots i_{n-1}ji_{n+1} \cdots i_N} = \sum_{i_n=1}^{I_n} a_{i_1 \cdots i_{n-1} i_n i_{n+1} \cdots i_N} U_{i_nj}.
\]
where $(\mathcal{A} \times_n \mathbf{U})_{i_1 \cdots i_{n-1}ji_{n+1} \cdots i_N}$ means the $(i_1, \cdots, i_{n-1}, j, i_{n+1}, \cdots, i_N)$-th element of $\mathcal{A} \times_n \mathbf{U}$, \( a_{i_1i_2 \cdots i_n \cdots i_N} \) and \( U_{i_nj} \) are the elements of \( \mathcal{A} \) and \( \mathbf{U} \) respectively.

As the tensor notation is adopted, the optimization problem of PSA can be expressed in the following form
\[
    \left\{
    \begin{aligned}
        \max        & \quad \mathcal{S} \times_1 \mathbf{w} \times_2 \mathbf{w} \times_3 \mathbf{w} \\
        \text{s.t.} & \quad \mathbf{w}^{\mathrm{T}} \mathbf{w} = 1
    \end{aligned}
    \right..
\]
Then the problem is turned into a tensor eigenvector problem, and the iterative formula of PSA can be obtained as
\begin{equation} \label{eq:PSAFormular}
    \mathbf{w} \leftarrow \mathcal{S} \times_2 \mathbf{w} \times_3 \mathbf{w}, \quad \mathbf{w} \leftarrow \frac{\mathbf{w}}{||\mathbf{w}||}.
\end{equation}
Similar to FastICA, the orthogonal constraint is enforced during the iterative process to ensure that the unmixing vectors will not converge into each other in PSA. Note that the iterative formula of PSA \cref{eq:PSAFormular} only involves the calculation of the third-order coskewness tensor of the data, which makes it more efficient than FastICA when the number of samples \( L \) is large. Moreover, if we adopt the skewness as the non-gaussianity metric function in FastICA, the results of FastICA and PSA will be exactly the same. Therefore, in practical applications, PSA is preferred over FastICA when dealing with source signals that exhibit high skewness.

\subsection{JADE}

The two algorithms aforementioned are only aimed at real signals, while the JADE algorithm is able to process complex signals and select fourth-order cumulants as the non-gaussianity metric. Assume that there exist four complex vectors \( \mathbf{x} = \begin{bmatrix} x_1 & x_2 & \cdots & x_L \end{bmatrix}^{\mathrm{H}} \), \( \mathbf{y} = \begin{bmatrix} y_1 & y_2 & \cdots & y_L \end{bmatrix}^{\mathrm{H}} \), \( \mathbf{z} = \begin{bmatrix} z_1 & z_2 & \cdots & z_L \end{bmatrix}^{\mathrm{H}} \) and \( \mathbf{w} = \begin{bmatrix} w_1 & w_2 & \cdots & w_L \end{bmatrix}^{\mathrm{H}} \), where $\cdot^{\mathrm{H}}$ means conjugate transpose, then the fourth-order cumulants of \( \mathbf{x} \), \( \mathbf{y} \), \( \mathbf{z} \) and \( \mathbf{w} \) can be calculated as
\[
    \begin{split}
        \text{cum}(\mathbf{x}, \mathbf{y}^*, \mathbf{z}, \mathbf{w}^*) & = \frac{1}{L}\sum_{m=1}^{L}x_m y_m^* z_m w_m^* - \frac{1}{L^2}\sum_{m=1}^{L}x_m y_m^*\sum_{m=1}^{L}z_m w_m^*                    \\
                                                                       & -\frac{1}{L^2}\sum_{m=1}^{L}x_m z_m\sum_{m=1}^{L}y_m^{*} w_m^{*} - \frac{1}{L^2}\sum_{m=1}^{L}x_m w_m^*\sum_{m=1}^{L}z_m y_m^*,
    \end{split}
\]
where \( * \) denotes conjugation.

Another difference between JADE and the other two algorithms is that JADE simultaneously estimates all the unmixing vectors as a time. Let \( \mathbf{W} = \begin{bmatrix} \mathbf{w}_1 & \mathbf{w}_2 & \cdots & \mathbf{w}_N \end{bmatrix} \) be the matrix that consists of all wanted unmixing vectors, then the optimization function of JADE can be formulated as follows, noting that the data matrix $\mathbf{X}$ is complex,
\begin{equation}\label{eq:JADEFormular}
    \begin{cases}
        \max        & \quad \sum_{i,k,l = 1}^{N} |\text{cum}(\mathbf{w}_i^{\mathrm{H}}\mathbf{X}, \mathbf{w}_i^{\mathrm{H}}\mathbf{X}^*, \mathbf{w}_k^{\mathrm{H}}\mathbf{X}, \mathbf{w}_l^{\mathrm{H}}\mathbf{X}^*)|^2 \\
        \text{s.t.} & \quad \mathbf{W}^{\mathrm{H}}\mathbf{W} = \mathbf{I}
    \end{cases}.
\end{equation}

Cardoso proved that the optimization problem above was equivalent to a joint matrix diagonalization problem as follows
\begin{equation}
    \begin{cases}
        \max        & \quad \sum_{k, l=1}^{K} \lVert \text{diag}(\mathbf{W} \mathbf{N}_{kl} \mathbf{W}^{\mathrm{H}}) \rVert^2 \\
        \text{s.t.} & \quad \mathbf{W}^{\mathrm{H}}\mathbf{W} = \mathbf{I}
    \end{cases} ,
\end{equation}
where $\text{diag}(\cdot)$ means extracting the diagonal elements of the matrix to form a new vector and
\[
    (\mathbf{N}_{kl})_{ij} = \text{cum}(X_i, X_j^*,X_k, X_l^*),
\]
where $X_i$, $X_j$, $X_k$ and $X_l$ are the $i$-th, $j$-th, $k$-th and $l$-th row vectors of $\mathbf{X}$, respectively. Moreover, the unmixing matrix $\mathbf{W}$ can be obtained via joint matrix diagonalization algorithm, such as successive Givens rotation \cite{cardoso1993blind}.

%% file: tex/4-method.tex
\section{Method}

Inspired by PSA, we find that the methods based on fourth-order statistics can also be contributed as a high-order tensor eigenpair problem. Hence, in this section, we first derive the basic principles of Principal Kurtosis Analysis (PKA), which can be considered as a generalization of PSA. Then, we notice that the eigenvectors of a high-order tensor are generally unorthogonal, which means that the orthogonal constraint used in previously mentioned methods could result in a suboptimal solution. To overcome this problem, we propose a brand new non-zero volume constraint. Further, in order to accelerate the convergence speed and improve the accuracy of the solution, we introduce the Riemann gradient.

\subsection{Principal Kurtosis Analysis}
As well known, the kurtosis of a vector $\mathbf{x}=\begin{bmatrix} x_1 & x_2 & ... & x_L \end{bmatrix}^{\mathrm{T}}$ is defined as follows:
\begin{equation}\label{eq:kurtosis}
    \text{kurt}(\mathbf{x}) = \frac{1}{L}\sum_{i=1}^{L}\frac{(x_{i} - \bar{x})^4}{\sigma^4} ,
\end{equation}
where $\bar{x}$ is the mean value of $\mathbf{x}$ and $\sigma$ is the standard variance. If the signal is pre-whitened, i.e., $\bar{x}=0$ and $\sigma^2=1$, then \cref{eq:kurtosis} can be simplified to
\begin{equation}
    \text{kurt}(\mathbf{x}) = \frac{1}{L}\sum_{i=1}^{L}x_{i}^4.
\end{equation}
For a complex vector, it can be written as
\begin{equation}
    \text{kurt}(\mathbf{x}) = \frac{1}{L}\sum_{i=1}^{L}|x_{i}|^4.
\end{equation}
Compared with the third-order statistics, the fourth-order statistics (kurtosis) is a much more popular choice for dealing with complex data. Moreover, through a thorough investigation, we find that ICA methods based on kurtosis can be attributed to a high-order tensor eigenpair problem. The details can be seen in the following theorem.

\begin{theorem} \label{thm:kurtTensor}
    Supposing $\mathbf{X}=\begin{bmatrix} \mathbf{x}_1 & \mathbf{x}_2 & ... & \mathbf{x}_L \end{bmatrix} \in \mathbb{C}^{N \times L}$, which is a pre-whitened complex data matrix, then the fourth-order statistics of $\mathbf{X}$ is completely determined by a fourth-order tensor $\mathcal{C} \in \mathbb{C}^{N \times N \times N \times N}$. More specifically, for a complex vector $\mathbf{w} \in \mathbb{C}^{N}$, the kurtosis of the projected data $\mathbf{w}^{\mathrm{H}}\mathbf{X}$ is $\mathrm{kurt}(\mathbf{w}^{\mathrm{H}}\mathbf{X})=\mathcal{C} \times_1 \mathbf{w} \times_2 \mathbf{w}^{*} \times_3 \mathbf{w} \times_4 \mathbf{w}^{*}$.
\end{theorem}
\begin{proof}
    It should be noted that $\mathrm{kurt}(\mathbf{w}^{\mathrm{H}}\mathbf{X})$ can be written as the following summation form
    \begin{equation}
        \begin{aligned}
            \mathrm{kurt}(\mathbf{w}^{\mathrm{H}}\mathbf{X}) = & \frac{1}{L}\sum_{j=1}^{L}|\mathbf{w}^{\mathrm{H}}\mathbf{x}_j|^4                                                                                                                                                                  \\
            =                                                  & \frac{1}{L}\sum_{j=1}^{L}\left(\mathbf{w}^{\mathrm{H}}\mathbf{x}_j\right)\left(\mathbf{w}^{\mathrm{H}}\mathbf{x}_j\right)^{*}\left(\mathbf{w}^{\mathrm{H}}\mathbf{x}_j\right)\left(\mathbf{w}^{\mathrm{H}}\mathbf{x}_j\right)^{*} \\
            =                                                  & \frac{1}{L}\sum_{j=1}^{L}\left(\sum_{i_1=1}^{N}x_{ji_1}w_{i_1}\right)\left(\sum_{i_2=1}^{N}x_{ji_2}w_{i_2}\right)^{*}\left(\sum_{i_3=1}^{N}x_{ji_3}w_{i_3}\right)\left(\sum_{i_4=1}^{N}x_{ji_4}w_{i_4}\right)^{*}
        \end{aligned}.
    \end{equation}
    By rearranging the summation sign, we have
    \begin{equation}
        \begin{aligned}
            \mathrm{kurt}(\mathbf{w}^{\mathrm{H}}\mathbf{X}) = & \frac{1}{L}\sum_{j=1}^{L}\sum_{i_1, i_2, i_3, i_4=1}^{N}x_{ji_1}x_{ji_2}^{*}x_{ji_3}x_{ji_4}^{*}w_{i_1}w_{i_2}^{*}w_{i_3}w_{i_4}^{*}            \\
            =                                                  & \sum_{i_1, i_2, i_3, i_4=1}^N\left(\frac{1}{L}\sum_{j=1}^{L}x_{ji_1}x_{ji_2}^{*}x_{ji_3}x_{ji_4}^{*}\right)w_{i_1}w_{i_2}^{*}w_{i_3}w_{i_4}^{*} \\
        \end{aligned}.
    \end{equation}
    Let \( c_{i_1 i_2 i_3 i_4} = \frac{1}{L}\sum_{j=1}^{L}x_{ji_1}x_{ji_2}^{*}x_{ji_3}x_{ji_4}^{*} \) be the $(i_1, i_2, i_3, i_4)$-th element of the tensor $\mathcal{C}$, then we can simplify the equation above by using the $n$-mode product, i.e.,
    \begin{equation}
        \begin{aligned}
            \mathrm{kurt}(\mathbf{w}^{\mathrm{H}}\mathbf{X}) = & \sum_{i_1, i_2, i_3, i_4=1}^{N}c_{i_1i_2i_3i_4}w_{i_1}w_{i_2}^{*}w_{i_3}w_{i_4}^{*}                                                                         \\
            =                                                  & \sum_{i_4=1}^{N}\left(\sum_{i_3=1}^{N}\left(\sum_{i_2=1}^{N}\left(\sum_{i_1=1}^{N}c_{i_1i_2i_3i_4}w_{i_1}\right)w_{i_2}^{*}\right)w_{i_3}\right)w_{i_4}^{*} \\
            =                                                  & \mathcal{C} \times_1 \mathbf{w} \times_2 \mathbf{w}^{*} \times_3 \mathbf{w} \times_4 \mathbf{w}^{*}
        \end{aligned}.
    \end{equation}
    Moreover, the tensor $\mathcal{C}$ is a symmetric tensor, whose elements satisfy the following equations
    \[
        \begin{aligned}
            c_{i_1i_2i_3i_4} & = c_{i_2i_1i_3i_4}^* = c_{i_3i_2i_1i_4} = c_{i_4i_2i_3i_1}^* \\
                             & = c_{i_1i_3i_2i_4}^* = c_{i_1i_4i_3i_2}                      \\
                             & = c_{i_1i_2i_4i_3}^*.
        \end{aligned}
    \]
\end{proof}

For convenience, denote $\mathcal{C} \times_1 \mathbf{w} \times_2 \mathbf{w}^{*} \times_3 \mathbf{w} \times_4 \mathbf{w}^{*}$ as $\mathcal{C}\mathbf{w}^4$, then the optimization problem can be formulated as
\begin{equation}\label{eq:kurtOptProblem}
    \begin{cases}
        \max        & \mathcal{C}\mathbf{w}^4             \\
        \text{s.t.} & \mathbf{w}^{\mathrm{H}}\mathbf{w}=1
    \end{cases}.
\end{equation}
Such a constrained optimization problem can be solved by the Lagrange multiplier method. We can construct the Lagrangian function for the optimization problem \cref{eq:kurtOptProblem} as follows
\[
    \mathcal{L} = \frac{1}{4}\mathcal{C}\mathbf{w}^4 - \frac{\lambda}{2}(\mathbf{w}^{\mathrm{H}}\mathbf{w}-1),
\]
where $\lambda$ is the Lagrangian coefficient. For the solution of the optimization problem, we have \cref{thm:eigenvector}.

\begin{theorem}\label{thm:eigenvector}
    The eigenvectors of the tensor $\mathcal{C}$ are the stationary points of the optimization problem \cref{eq:kurtOptProblem}.
\end{theorem}
\begin{proof}
    It should be noted that \( \mathcal{C}\mathbf{w}^4  \) has the following four equivalent expressions:
    \[
        \begin{aligned}
            \mathcal{C}\mathbf{w}^4 & = (\mathcal{C} \times_2 \mathbf{w}^* \times_3 \mathbf{w} \times_4 \mathbf{w}^*) \times_1 \mathbf{w}  \\
                                    & = (\mathcal{C} \times_1 \mathbf{w} \times_3 \mathbf{w} \times_4 \mathbf{w}^*) \times_2 \mathbf{w}^*  \\
                                    & = (\mathcal{C} \times_1 \mathbf{w} \times_2 \mathbf{w}^* \times_4 \mathbf{w}^*) \times_3 \mathbf{w}  \\
                                    & = (\mathcal{C} \times_1 \mathbf{w} \times_2 \mathbf{w}^* \times_3 \mathbf{w}) \times_4 \mathbf{w}^*.
        \end{aligned}
    \]
    Therefore, the derivative of \( \mathcal{C}\mathbf{w}^4 \) with respect to \( \mathbf{w} \) is
    \begin{equation}
        \begin{split}
            \frac{\partial \mathcal{C}\mathbf{w}^4}{\partial \mathbf{w}} & = \operatorname{vec}\left( \mathcal{C}\times_2 \mathbf{w}^{*} \times_3 \mathbf{w} \times_4 \mathbf{w}^* \right) \\
                                                                         & + \operatorname{vec}\left( \mathcal{C}\times_1 \mathbf{w} \times_3 \mathbf{w} \times_4 \mathbf{w}^* \right)^*   \\
                                                                         & + \operatorname{vec}\left( \mathcal{C}\times_1 \mathbf{w} \times_2 \mathbf{w}^* \times_4 \mathbf{w}^* \right)   \\
                                                                         & + \operatorname{vec}\left( \mathcal{C}\times_1 \mathbf{w} \times_2 \mathbf{w}^* \times_3 \mathbf{w} \right)^*.
        \end{split}
    \end{equation}
    where the operator $\operatorname{vec}(\cdot)$ is defined as the vectorization of a tensor.
    Moreover, according to the symmetries of the tensor \( \mathcal{C} \), we have that the $j$-th element of the two vectors \( \operatorname{vec}\left( \mathcal{C}\times_2 \mathbf{w}^{*} \times_3 \mathbf{w} \times_4 \mathbf{w}^* \right) \) and \( \operatorname{vec}\left( \mathcal{C}\times_1 \mathbf{w} \times_3 \mathbf{w} \times_4 \mathbf{w}^* \right) \) has the following relationship
    \[
        \begin{split}
              & \operatorname{vec}\left( \mathcal{C}\times_2 \mathbf{w}^{*} \times_3 \mathbf{w} \times_4 \mathbf{w}^* \right)_j = \sum_{i_2, i_3, i_4}^{N} c_{ji_2i_3i_4}w_{i_2}^{*}w_{i_3}w_{i_4}^{*} \\
            = & \left( \sum_{i_2, i_3, i_4}^{N} c_{ji_2i_3i_4}^* w_{i_2}w_{i_3}^*w_{i_4} \right)^*   =  \left( \sum_{i_1, i_3, i_4}^{N} c_{i_1ji_3i_4} w_{i_1}w_{i_3} w_{i_4}^*\right)^*               \\
            = & \operatorname{vec}\left( \mathcal{C}\times_1 \mathbf{w} \times_3 \mathbf{w} \times_4 \mathbf{w}^* \right)_j^*.
        \end{split}
    \]
    Therefore, it can be concluded that
    \[
        \operatorname{vec}\left( \mathcal{C}\times_2 \mathbf{w}^{*} \times_3 \mathbf{w} \times_4 \mathbf{w}^* \right) = \operatorname{vec}\left( \mathcal{C}\times_1 \mathbf{w} \times_3 \mathbf{w} \times_4 \mathbf{w}^* \right)^*.
    \]
    Similarly, we have
    \[
        \begin{split}
            \operatorname{vec}\left( \mathcal{C}\times_2 \mathbf{w}^{*} \times_3 \mathbf{w} \times_4 \mathbf{w}^* \right) & = \operatorname{vec}\left( \mathcal{C}\times_1 \mathbf{w} \times_2 \mathbf{w}^* \times_4 \mathbf{w}^* \right), \\
            \operatorname{vec}\left( \mathcal{C}\times_2 \mathbf{w}^{*} \times_3 \mathbf{w} \times_4 \mathbf{w}^* \right) & = \operatorname{vec}\left( \mathcal{C}\times_1 \mathbf{w} \times_2 \mathbf{w}^* \times_3 \mathbf{w} \right)^*.
        \end{split}
    \]
    For simplicity, we denote \( \operatorname{vec}(\mathcal{C}\times_2 \mathbf{w}^{*} \times_3 \mathbf{w} \times_4 \mathbf{w}^* )\) as \( \mathcal{C}\mathbf{w}^3 \), then the derivative of \( \mathcal{C}\mathbf{w}^4 \) with respect to \( \mathbf{w} \) can be written as
    \[
        \frac{\partial \mathcal{C}\mathbf{w}^4}{\partial \mathbf{w}} = 4\mathcal{C}\mathbf{w}^3.
    \]
    Hence, the derivative of \( \mathcal{L} \) with respect to \( \mathbf{w} \) can be written as
    \[
        \frac{\partial \mathcal{L}}{\partial \mathbf{w}} = \mathcal{C}\mathbf{w}^3 - \lambda \mathbf{w}.
    \]
    Setting the derivative equal to zero, we can get the stationary point of the optimization problem \cref{eq:kurtOptProblem} as
    \[
        \mathcal{C}\mathbf{w}^3 = \lambda \mathbf{w}.
    \]
    Namely, $\mathbf{w}$ is the eigenvector of the tensor $\mathcal{C}$.
\end{proof}

From \cref{thm:eigenvector}, we can find that the eigenvectors of the tensor $\mathcal{C}$ can be simply obtained by the fixed-point iteration method, i.e.,
\begin{equation}
    \label{eq:orgFormula}
    \mathbf{w} \leftarrow \frac{\mathcal{C}\mathbf{w}^{3}}{\left\| \mathcal{C}\mathbf{w}^{3} \right\|}.
\end{equation}
Assume that \( \mathbf{W}_{k-1} = \begin{bmatrix} \mathbf{w}_1 & \mathbf{w}_2 & \cdots & \mathbf{w}_{k-1} \end{bmatrix} \) is the matrix consisting the first \( k-1 \) obtained unmixing vectors, then the iteration formula for the \( k \)-th unmixing vector can be written as
\begin{equation}\label{eq:iterFormula}
    \mathbf{w}_{k} \leftarrow \frac{\mathbf{P}_{\mathbf{W}_{k-1}} \mathcal{C}\mathbf{w}_{k}^{3}}{\left\| \mathbf{P}_{\mathbf{W}_{k-1}}\mathcal{C}\mathbf{w}_{k}^{3} \right\|}.
\end{equation}
where \( \mathbf{P}_{\mathbf{W}_{k-1}} = \mathbf{I} - \mathbf{W}_{k-1} \mathbf{W}_{k-1}^{\mathrm{T}} \) is the orthogonal complement projection matrix. However, as we have mentioned before, the eigenvectors of a high-order tensor could be unorthogonal, which means that the orthogonal constraint used in the above formula may result in a suboptimal solution. To overcome this problem, we propose a brand new non-zero volume constraint, which will be introduced in the next subsection.

\subsection{Non-zero volume constraint}

Traditional algorithms use orthogonal constraints to prevent convergence to the same point, e.g., FastICA performs orthogonalization after every iteration. Under the sources' independence assumption, these algorithms can obtain the exact solution because the unmixing vectors are orthogonal. However, in some particular cases, e.g., radar anti-interference, source signals are correlated which means unmixing vectors are nonorthogonal. In such cases, orthogonal constraints fail to get the exact solution. Some algorithms use nonorthogonal constraints, such as NPSA \cite{geng2020npsa}, but the author indicated that although the solutions are nonorthogonal, the solution is also imprecise. To overcome the problem, we propose a new nonorthogonal constraint, namely the non-zero volume constraint, to successively obtain the exact solution while maintaining the solution's uniqueness.

Note that the \( k \) unmixing vectors in \( \mathbf{W}_k \) form a vertex in an \( N \) dimensional space. Since these \( k \) vectors are different from each other, the volume of the vertex spanned by these vectors, whose value equals the square root of the determinant of the covariance matrix $\mathbf{R}_k = \mathbf{W}_k^{\mathrm{H}}\mathbf{W}_k$, is non-zero. Therefore, to prevent the \( k \)-th unmixing vector from converging to the same solution as the previous ones, we introduce the non-zero volume constraint, and the optimization problem for solving the \( k \)-th unmixing vector can be formulated as

\begin{equation}
    \label{eq:KurtOpt2}
    \begin{cases}
        \max        & \mathcal{C}\mathbf{w}_{k}^4                   \\
        \text{s.t.} & \mathbf{w}_{k}^{\mathrm{H}}\mathbf{w}_{k} = 1 \\
                    & \text{det}(\mathbf{R}_{k}) \neq 0
    \end{cases} .
\end{equation}
Notice that after introducing the non-zero volume constraint, it's hard to construct an iterative formula like \cref{eq:iterFormula} to solve the optimization problem. Fortunately, the optimization problem \cref{eq:KurtOpt2} still can be solved by the gradient ascent method. The iterative updating formula of the $k$-th unmixing vector $\mathbf{w}_k$ can be written as
\begin{equation}
    \label{eq:optFormula}
    \mathbf{w}_k \leftarrow \mathbf{w}_k + \alpha \frac{\mathcal{C}\mathbf{w}_k^3}{|\text{det}(\mathbf{R}_k)|}, \quad \mathbf{w}_k \leftarrow \frac{\mathbf{w}_k}{\lVert \mathbf{w}_k \rVert} ,
\end{equation}
where $\text{det}(\cdot)$ means the determinant of a matrix.

Furthermore, we have \cref{thm:equivalent} that ensures the equivalence of the stationary points corresponding to \cref{eq:iterFormula} and \cref{eq:optFormula}.
\begin{theorem}\label{thm:equivalent}
    The eigenvectors of a fourth-order tensor $\mathcal{C}$ are the stationary points under the gradient ascent iteration formula \cref{eq:optFormula}
\end{theorem}

\begin{proof}
    From \cref{eq:orgFormula}, it can be indicated that the stationary point \( \hat{\mathbf{w}} \) of the original target optimizing function statisfies the following equation
    \[
        \mathcal{C}\hat{\mathbf{w}}^{3} = \lambda \hat{\mathbf{w}},
    \]
    i.e., \( \hat{\mathbf{w}} \) is the eigenvector of the tensor \( \mathcal{C} \). Substituting the stationary point \( \hat{\mathbf{w}} \) into the updating formula \cref{eq:optFormula}, we have
    \begin{equation}\label{eq:proof1}
        \begin{split}
            \frac{\hat{\mathbf{w}} + \alpha \frac{\mathcal{C}\hat{\mathbf{w}}^3}{\left|\text{det}(\hat{\mathbf{R}})\right|}}{\left\| \hat{\mathbf{w}} + \alpha \frac{\mathcal{C}\hat{\mathbf{w}}^3}{\left|\text{det}(\hat{\mathbf{R}})\right|} \right\|}
            = \frac{\hat{\mathbf{w}} + \alpha \frac{\lambda \hat{\mathbf{w}}}{\left|\text{det}(\hat{\mathbf{R}})\right|}}{\left\| \hat{\mathbf{w}} + \alpha \frac{\lambda \hat{\mathbf{w}}}{\left|\text{det}(\hat{\mathbf{R}})\right|} \right\|}
            = \frac{\left( 1 + \frac{\alpha \lambda}{\left|\det \left( \hat{\mathbf{R}} \right)\right|} \right) \hat{\mathbf{w}}}{\left( 1 + \frac{\alpha \lambda}{\left|\det \left( \hat{\mathbf{R}} \right)\right|} \right) \left\| \hat{\mathbf{w}} \right\|} = \hat{\mathbf{w}},
        \end{split}
    \end{equation}
    where \( \hat{\mathbf{R}} \) is the covariance matrix of \( \hat{\mathbf{w}} \) and the previous obtained unmixing vectors. From \cref{eq:proof1}, it can be found that the stationary point \( \hat{\mathbf{w}} \) of the original target optimizing function is the stationary point of the updating formula \cref{eq:optFormula}.

    Similarly, it can be indicated that the stationary point $\hat{\mathbf{w}}$ of \cref{eq:optFormula} satisfies
    \[
        \frac{\mathcal{C}\hat{\mathbf{w}}^3}{\left|\det\left( \hat{\mathbf{R}} \right)\right|} = \lambda \hat{\mathbf{w}}.
    \]
    Multiplying $\left|\det\left( \hat{\mathbf{R}} \right)\right|$ on the both sides, it can be formulated as
    \[
        \mathcal{C}\hat{\mathbf{w}}^3 = \lambda\left|\det\left( \hat{\mathbf{R}} \right)\right| \hat{\mathbf{w}}.
    \]
    Since $\lambda$ is a constant, the constant item $\left|\det\left( \hat{\mathbf{R}} \right)\right|$ can be absorbed into $\lambda$, i.e.,
    \begin{equation}
        \mathcal{C}\hat{\mathbf{w}}^3 = \lambda\hat{\mathbf{w}}.
    \end{equation}
    Substituting the equation above into \cref{eq:orgFormula}, we have
    \begin{equation}\label{eq:proof2}
        \frac{\mathcal{C}\hat{\mathbf{w}}^{3}}{\left\| \mathcal{C}\hat{\mathbf{w}}^{3} \right\|} = \frac{\lambda \hat{\mathbf{w}}}{\left\| \lambda \hat{\mathbf{w}} \right\|} = \hat{\mathbf{w}}.
    \end{equation}
    From \cref{eq:proof2}, it can be seen that the stationary point of \cref{eq:optFormula} is also the stationary point of \cref{eq:orgFormula}. Therefore, the stationary points corresponding to the two formulas are equivalent.
\end{proof}

According to \cref{eq:optFormula}, it can be found that \( \mathbf{w}_k \) will not converge to the same solution as the previous ones, which guarantees the uniqueness of the solution. Meanwhile, \cref{thm:equivalent} ensures that the solution obtained by the non-zero volume constraint is an eigenvector of the tensor $\mathcal{C}$. Hence, it can be concluded that the non-zero volume constraint is a more suitable constraint than the orthogonal constraint for the ICA problem when the source signals are correlated.

\subsection{Riemann gradient}

Although the introduction of the non-zero volume constraint ensures the uniqueness of the solution, it could result in a slow convergence speed. To better illustrate the difference between the normal gradient and the Riemann gradient, we denote $\mathbf{w}_k^{(t)}$ as the $t$-th iteration result of $\mathbf{w}_k$. If $\mathbf{w}_k^{(t)}$ is close to one of the previous unmixing vectors $\mathbf{w}_l$, the determinant of $\mathbf{R}_k$ tends towards $0$, which makes the value of $\frac{1}{|\text{det}(\mathbf{R}_k)|}$ very huge. Hence, in this case, we have that
\[
    \mathbf{w}_k^{(t)} + \alpha \frac{\mathcal{C}\mathbf{w}_k^{(t)^3}}{|\text{det}(\mathbf{R}_k)|} \approx \alpha \frac{\mathcal{C}\mathbf{w}_k^{(t)^3}}{|\text{det}(\mathbf{R}_k)|}.
\]
Therefore, after normalization, $\mathbf{w}_k^{(t+1)}$ can be written as
\[
    \mathbf{w}_k^{(t+1)} \approx \frac{\mathcal{C}\mathbf{w}_k^{(t)^{3}}}{\left\| \mathcal{C}\mathbf{w}_k^{(t)^{3}} \right\|},
\]
which is exactly the same updating formula as \cref{eq:orgFormula}. Since \( \mathbf{w}_k^{(t)} \) is close to a previously obtained unmixing vector \( \mathbf{w}_l \), it can be indicated that under the updating formula of \cref{eq:optFormula}, the converged solution of \( \mathbf{w}_k \) will be \( \mathbf{w}_l \). However, due to the non-zero volume constraint, \( \mathbf{w}_k \) will not converge into \( \mathbf{w}_l \). Thus, in this case, \( \mathbf{w}_k \) will gradually approach \( \mathbf{w}_l \), but never achieve. As can be seen in \cref{fig:normalUpdating}, when \( \mathbf{w}_k^{(t)} \) is close to one of the previous unmixing vectors, the gradient of \( \mathbf{w}_k^{(t)} \) will be almost parallel to \( \mathbf{w}_k^{(t)} \), which will cause excessive gradient values along \( \mathbf{w}_k^{(t)} \). Meanwhile, the direction perpendicular to \( \mathbf{w}_k^{(t)} \) is very small, which is actually the valid part of the gradient. After normalization, the valid part of the gradient of \( \mathbf{w}_k^{(t)} \) will be further reduced, which leads to a slow convergence speed.

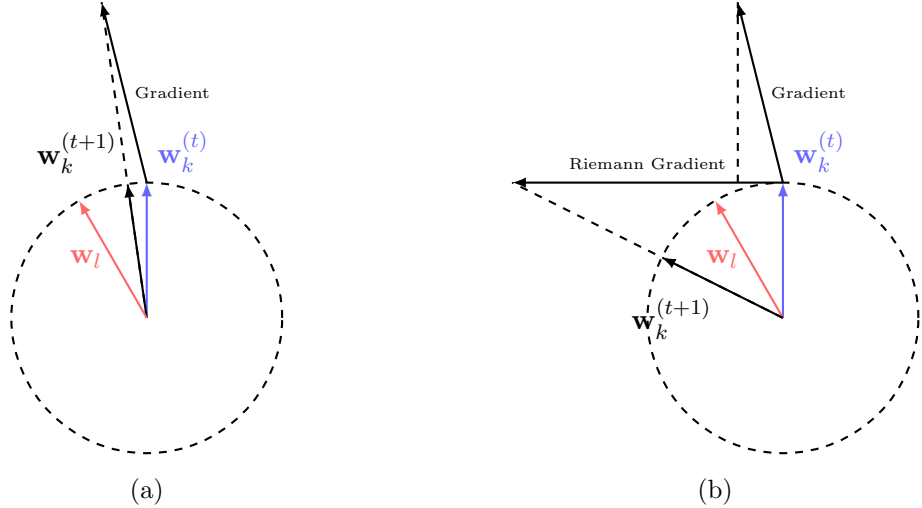
\begin{figure}[htb!]
    \centering
    \begin{subfigure}{.45\textwidth}
        \centering
        \begin{tikzpicture}[scale=0.6]
            \draw[thick, dashed] (0,0) circle (3);
            \draw[thick, dashed] (0,0) -- (-1, 7);
            \draw[thick, -latex, blue!60!white] (0,0) -- (0,3) node[above right] {$\mathbf{w}_k^{(t)}$};
            \draw[thick, -latex] (0,3) -- (-1, 7) node[midway,right] {\tiny Gradient};
            \draw[thick, -latex] (0,0) -- (-1/2.35, 7/2.35) node[above left] {$\mathbf{w}_k^{(t+1)}$};
            \draw[thick, -latex, red!60!white] (0,0) -- (-0.5*3, 0.866*3) node[midway, left] {$\mathbf{w}_l$};
        \end{tikzpicture}
        \caption{}
        \label{fig:normalUpdating}
    \end{subfigure}
    \begin{subfigure}{.45\textwidth}
        \centering
        \begin{tikzpicture}[scale=0.6]
            \draw[thick, dashed] (0,0) circle (3);
            \draw[thick, dashed] (0,0) -- (-6, 3);
            \draw[thick, -latex, blue!60!white] (0,0) -- (0,3) node[above right] {$\mathbf{w}_k^{(t)}$};
            \draw[thick, -latex] (0,3) -- (-1, 7) node[midway,right] {\tiny Gradient};
            \draw[thick, -latex] (0,3) -- (-6, 3) node[midway,above=1pt] {\tiny Riemann Gradient};
            \draw[thick, dashed] (-1,3) -- (-1, 7);
            \draw[thick, -latex, red!60!white] (0,0) -- (-0.5*3, 0.866*3) node[midway, left] {$\mathbf{w}_l$};
            \draw[thick, -latex] (0,0) -- (-6*0.447, 3*0.447) node[midway,below left] {$\mathbf{w}_k^{(t+1)}$};
        \end{tikzpicture}
        \caption{}
        \label{fig:RiemannUpdating}
    \end{subfigure}
    \caption{A simple illustration of the difference between the normal gradient and the Riemann gradient when $\mathbf{w}_k^{(t)}$ is close to a previously obtained unmixing vector  $\mathbf{w}_l$. (a) normal gradient (b) Riemann gradient}
    \label{fig_diff}
\end{figure}

To overcome this problem, we introduce the Riemann gradient. Noting that the optimization problem \cref{eq:KurtOpt2} is constrained on the Riemann manifold $\mathbf{w}_k^{\mathrm{H}}\mathbf{w}_k = 1$, hence, we can construct a projection operator mapping the tangent space at $\mathcal{C}\mathbf{w}_k^3$ to the manifold \cite{bonnabel2013stochastic}. Moreover, in \cite{bonnabel2013stochastic}, the author demonstrated that the gradient ascent method on such a Riemann manifold is both convergent and effective, as it can effectively mitigate the influence of irrelevant gradient components. Therefore, replacing the gradient $\mathcal{C}\mathbf{w}_k^3$ by the Riemann gradient $(\mathbf{I} - \mathbf{w}_k\mathbf{w}_k^{\mathrm{H}})\mathcal{C}\mathbf{w}_k^3$ will help $\mathbf{w}_k$ converge to the stationary point in the subspace and speed up the convergence process.

Combining the non-zero volume constraint and Riemann gradient, the revised updating formula of $\mathbf{w}_k$ can be written as
\begin{equation}
    \mathbf{w}_k \leftarrow \mathbf{w}_k + \alpha \frac{(\mathbf{I} - \mathbf{w}_k\mathbf{w}_k^{\mathrm{H}})\mathcal{C}\mathbf{w}_k^3}{|\text{det}(\mathbf{R}_k)|} ,\quad \mathbf{w}_k \leftarrow \frac{\mathbf{w}_k}{\lVert \mathbf{w}_k \rVert} .
\end{equation}

The difference between the normal gradient and the Riemann gradient can be seen in \cref{fig:RiemannUpdating}. Via Riemann gradient, the gradient component along $\mathbf{w}_k^{(t)}$ is suppressed, hence, the updating process can be more efficient. When $\mathbf{w}_k^{(t)}$ is close to $\mathbf{w}_l$, $\frac{1}{|\det(\mathbf{R}_k)|}$ could be very large and the following statement holds
\begin{equation}
    \mathbf{w}_k^{(t)} + \alpha \frac{(\mathbf{I} - \mathbf{w}_k^{(t)}\mathbf{w}_k^{{(t)}^{\mathrm{H}}})\mathcal{C}\mathbf{w}_k^{{(t)}^3}}{|\text{det}(\mathbf{R}_k)|} \approx \alpha \frac{(\mathbf{I} - \mathbf{w}_k^{(t)}\mathbf{w}_k^{{(t)}^{\mathrm{H}}})\mathcal{C}\mathbf{w}_k^{{(t)}^3}}{|\text{det}(\mathbf{R}_k)|}.
\end{equation}
After normalization, $\mathbf{w}_k^{(t+1)}$ can be formulated as
\begin{equation}
    \label{eq:RiemannApprox}
    \mathbf{w}_k^{(t+1)} \approx \frac{(\mathbf{I} - \mathbf{w}_k^{(t)}\mathbf{w}_k^{(t)^{\mathrm{H}}})\mathcal{C}\mathbf{w}_k^{(t)^3}}{\lVert (\mathbf{I} - \mathbf{w}_k^{(t)}\mathbf{w}_k^{(t)^{\mathrm{H}}})\mathcal{C}\mathbf{w}_k^{(t)^3} \rVert}.
\end{equation}
According to \cref{eq:RiemannApprox}, it can be seen that $\mathbf{w}_k^{(t+1)}$ will be mainly determined by the Riemann gradient. Moreover, since the unrelated component of the gradient, which is along $\mathbf{w}_k^{(t)}$, has been eliminated, the perpendicular component will be prominent. Therefore, $\mathbf{w}_k^{(t+1)}$ will be approximately perpendicular to $\mathbf{w}_k^{(t)}$, allowing it to stand out from the vicinity of $\mathbf{w}_l$.

In summary, the pseudocode of our method is shown in \cref{alg:1}.

\begin{algorithm}
    \caption{Principal Kurtosis Analysis}
    \label{alg:1}
    \begin{algorithmic}
        \REQUIRE the fourth-order statistics tensor $\mathcal{C}$, the number of source signals $N$
        \ENSURE unmixing matrix $\mathbf{W}$

        \WHILE {$k$ = 1 to $N$}

        \STATE Initialize $\mathbf{w}_k$ with random unit vector.

        \WHILE {stop conditions are not met}
        \STATE $\mathbf{W}_k = \begin{bmatrix} \mathbf{w}_1 & \mathbf{w}_2 & ... & \mathbf{w}_{k} \end{bmatrix}$,
        \STATE $\mathbf{R}_{k} = \mathbf{W}_k^{\mathrm{H}} \mathbf{W}_k$,

        \STATE $\mathbf{w}_k \leftarrow \mathbf{w}_k + \alpha \frac{(\mathbf{I} - \mathbf{w}_k\mathbf{w}_k^{\mathrm{H}})\mathcal{C}\mathbf{w}_k^3}{|\text{det}(\mathbf{R}_k)|}$,

        \STATE $\mathbf{w}_k \leftarrow \frac{\mathbf{w}_k}{\lVert \mathbf{w}_k \rVert}$,
        \ENDWHILE
        \ENDWHILE

        \STATE All unmixing vectors form the unmixing matrix $\mathbf{W} = \begin{bmatrix} \mathbf{w}_1 & \mathbf{w}_2 & ... & \mathbf{w}_{N} \end{bmatrix}$.
    \end{algorithmic}
\end{algorithm}

%% file: tex/5-experiments.tex
\section{Experiments}

\nocite{kitamura2016determined, kitamura2018determined, kitamura2020consistent, sawada2023multi}

In this section, we first conduct a validation experiment to verify the effectiveness of the proposed PKA algorithm in solving the eigenpairs of the fourth-order statistical tensor. Subsequently, we perform experiments using a variety of data types, including basic waves, sound recordings, and radar signals, to assess the overall effectiveness of our proposed method. Moreover, five commonly used algorithms are selected as the comparison algorithms, including JADE \cite{cardoso1993blind}, cFastICA \cite{Bingham2000A}, PSA \cite{geng2014principal}, NPSA \cite{geng2020npsa}, and RPSA \cite{geng4438906rpsa}. As the data matrix will be whitened before performing BSS, unmixing vectors are nonorthogonal only when source signals are nonorthogonal, hence, to test performance in scenarios involving nonorthogonal unmixing vectors, it suffices to utilize nonorthogonal source signals. In the first experiment, we validate whether the algorithms could solve the eigenpairs correctly. In the second experiment, we aim to demonstrate the capability of our method in separating nonorthogonal signals. In the third experiment, we evaluate the performance of the algorithms in the sound data separation task. In the fourth experiment, we evaluate the suppression performance of the algorithms in the radar interference-suppressing task.

It should be noted that PSA and NPSA only process real signals, therefore, when dealing with complex signals, only the real part of the signal is used. Regarding the experimental setup, cFastICA uses kurtosis as the contract function, RPSA employs the default settings and PKA employs a learning rate of $\alpha = 1 \times 10^{-2}$. Meanwhile, the other algorithms do not require any parameter adjustments. Considering that in real-world scenarios, the unmixed signals with minimal kurtosis are likely to be the source signals, it would be more appropriate to use the gradient descent method rather than the gradient ascent method for PKA in this situation. Moreover, unlike JADE, the outcomes of cFastICA, PSA, NPSA, RPSA, and PKA, are influenced by the initial starting points. To avoid some extreme situations, when the methods exhibit obvious anomalies, we will try to reinitialize them to obtain a stable result. All algorithms are implemented in MATLAB R2022b on a laptop with 16GB RAM, Intel(R) Core (TM) i7-11800H CPU, @2.30GHz.

\subsection{Validation Experiment}

In this experiment, we construct a series of random fourth-order statistical tensors, and then use the mentioned algorithms to solve the eigenvectors. Let \( \mathbf{w} \) be an estimated eigenvector, and \( \mathcal{C} \) be the fourth-order statistical tensor, the following cosine similarity is used to evaluate whether \( \bm{w} \) is an eigenvector of \( \mathcal{C} \) or not
\begin{equation}\label{eq:cosine_similarity}
    s(\mathbf{w}) = \frac{\mathbf{w}^{\mathrm{H}}(\mathcal{C}\mathbf{w}^3)}{\lVert \mathbf{w} \rVert \lVert \mathcal{C}\mathbf{w}^3 \rVert} = \frac{\mathcal{C}\mathbf{w}^4}{\lVert \mathcal{C}\mathbf{w}^3 \rVert}.
\end{equation}
It can be found from \cref{eq:cosine_similarity} that \( s(\mathbf{w}) \) calculate the cosine of the angle between \( \mathbf{w} \) and \( \mathcal{C}\mathbf{w}^3 \). When \( s(\mathbf{w}) = 1 \),  it indicates that \( \mathcal{C} \mathbf{w}^3 = \lambda \mathbf{w} \), therefore, \( \mathbf{w} \) is exactly an eigenvector of \( \mathcal{C} \). When \( s(\mathbf{w}) \) equals 0, it means that \( \mathcal{C} \mathbf{w}^3 \) and \( \mathbf{w} \) are orthogonal to each other, hence, \( \mathbf{w} \) is not an eigenvector of \( \mathcal{C} \). Generally, the closer \( s(\mathbf{w}) \) is to 1, the more accurate the eigenvector \( \mathbf{w} \) is.

A total of 100 random fourth-order statistical tensors with the size of \( 3 \times 3 \times 3 \times 3 \) are generated, resulting in 300 eigenvectors to be estimated. We select a threshold between 0 and 1, and consider the eigenvector \( \mathbf{w} \) to be successfully found when \( s(\mathbf{w}) \) is greater than the threshold. By varying the threshold, a curve can be plotted to show the number of successfully found eigenvectors. The results of this experiment are shown in \cref{fig:EigenVectorSuccess}. It should be noted that PSA, NPSA, and RPSA are not included in this experiment as they are not designed for solving the eigenvectors of the fourth-order statistical tensor.

\begin{figure}[htbp!]
    \centering
    \includegraphics[width=.4\textwidth]{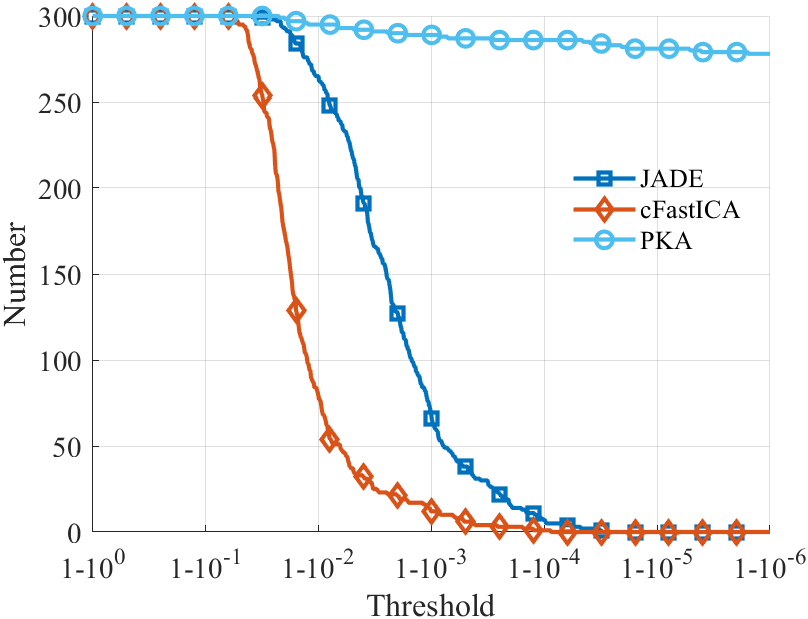}
    \caption{The variation of the number of successfully found eigenvectors with the threshold.}
    \label{fig:EigenVectorSuccess}
\end{figure}

From \cref{fig:EigenVectorSuccess}, it can be found as the threshold becomes closer to 1, the number of successfully found eigenvectors decreases for all algorithms. However, it can be observed that the curve of JADE and cFastICA rapidly drops to zero when the threshold is greater than \( 1-10^{-2} \). The reason is that the orthogonal constraint is imposed in JADE and cFastICA, which makes them unable to find the accurate solutions of the eigenvectors. In contrast, the curve of PKA is relatively stable, and most of the eigenvectors can be successfully found even when the threshold is greater than \( 1-10^{-6} \). Considering numerical errors, it can be concluded that the proposed PKA algorithm is able to solve the accurate eigenvectors of the fourth-order statistical tensor effectively.

\subsection{Experiment on Basic Waves}

In this experiment, we first generate a set of nonorthogonal basic waves and mix them with a random mixing matrix $\mathbf{A}$, then the mentioned algorithms are executed to separate the mixed signals. The source signals can be seen in \cref{fig:sine_wave_source}, where the first two signals are sine waves with frequencies 9Hz and 9.5Hz, respectively, and the third signal is a square wave with a frequency of 8Hz. Further, the randomly mixed signals can be seen in \cref{fig:sine_wave_mixed}. The covariance matrix of the source signals is
\[
    \begin{bmatrix}
        1    & 0.64 & 0.03 \\
        0.64 & 1    & 0.11 \\
        0.03 & 0.11 & 1
    \end{bmatrix},
\]
from which we can find that all three source signals are nonorthogonal to each other.

\begin{figure}[htb!]
    \centering
    \begin{subfigure}{.3\textwidth}
        \centering
        \includegraphics[width=\textwidth]{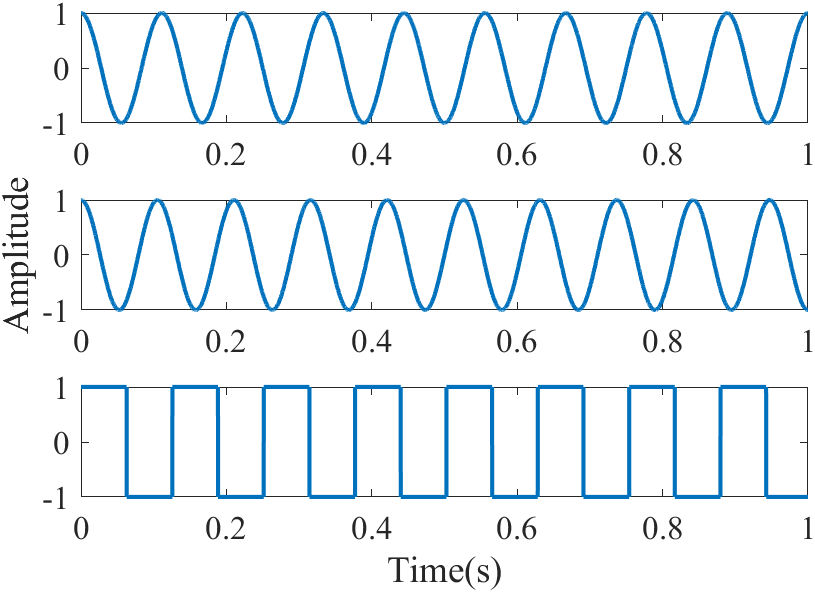}
        \caption{}
        \label{fig:sine_wave_source}
    \end{subfigure}
    \begin{subfigure}{.3\textwidth}
        \centering
        \includegraphics[width=\textwidth]{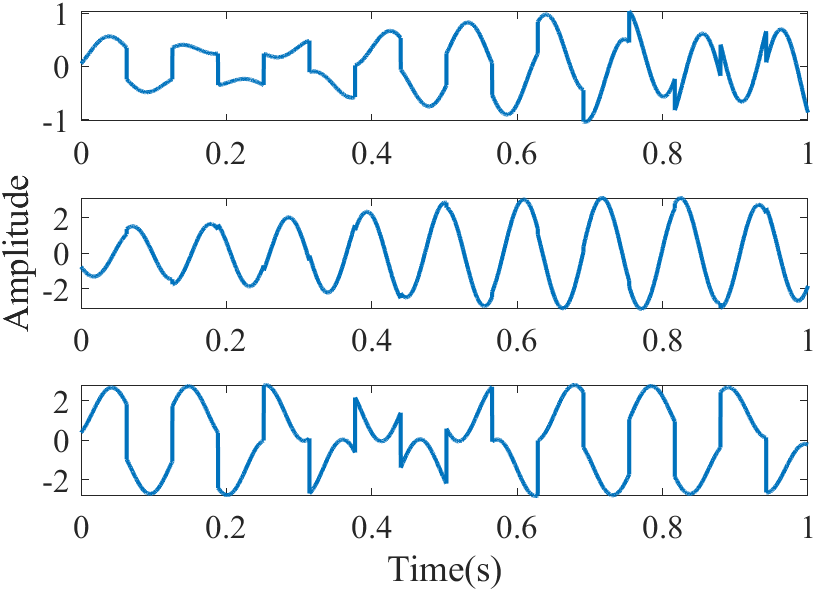}
        \caption{}
        \label{fig:sine_wave_mixed}
    \end{subfigure}

    \begin{subfigure}{.3\textwidth}
        \centering
        \includegraphics[width=\textwidth]{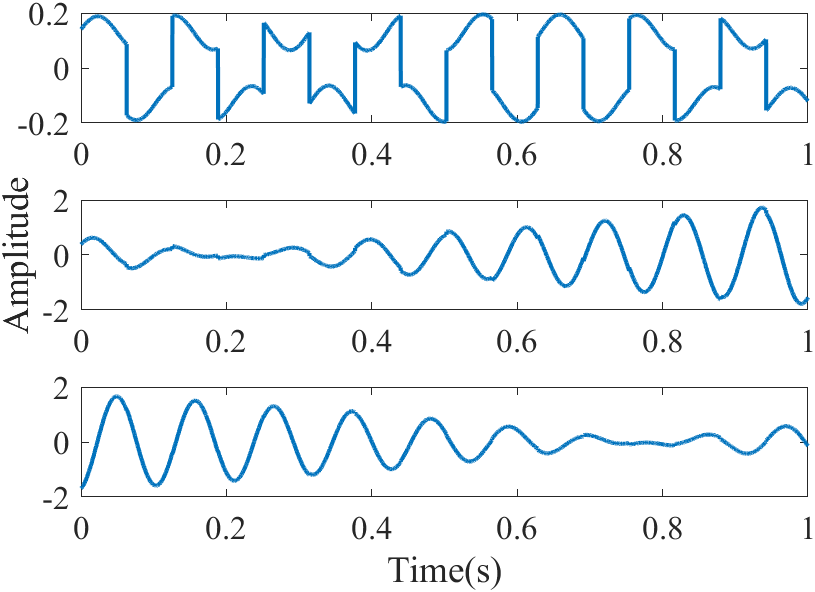}
        \caption{}
        \label{fig:sine_wave_JADE}
    \end{subfigure}
    \begin{subfigure}{.3\textwidth}
        \centering
        \includegraphics[width=\textwidth]{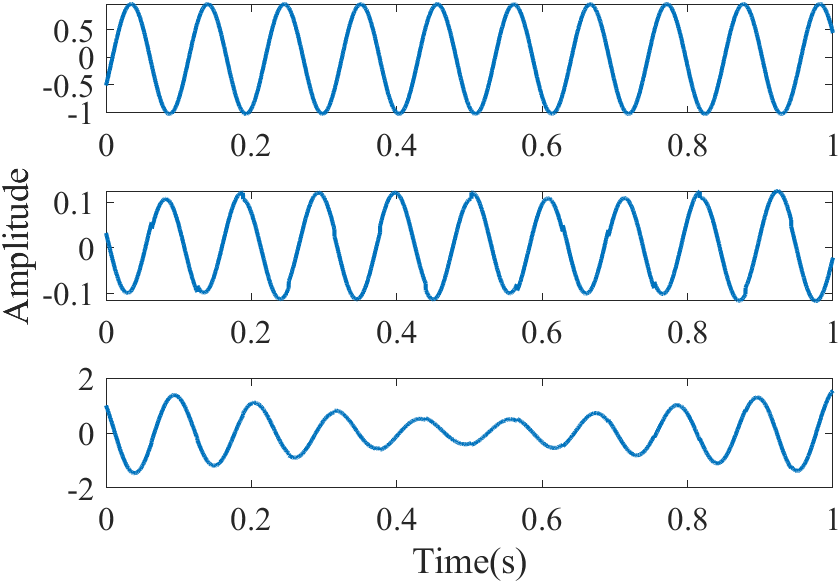}
        \caption{}
        \label{fig:sine_wave_cFastICA}
    \end{subfigure}
    \begin{subfigure}{.3\textwidth}
        \centering
        \includegraphics[width=\textwidth]{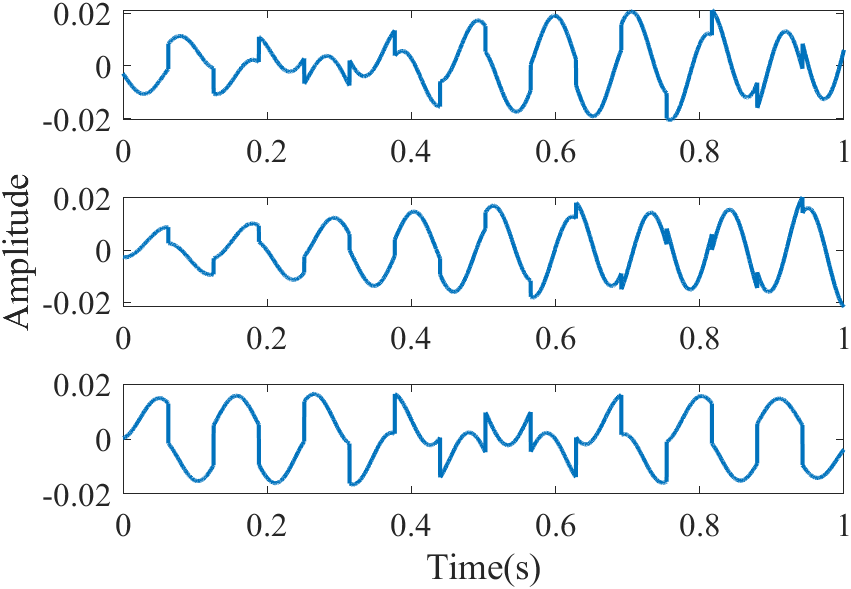}
        \caption{}
        \label{fig:sine_wave_PSA}
    \end{subfigure}
    \begin{subfigure}{.3\textwidth}
        \centering
        \includegraphics[width=\textwidth]{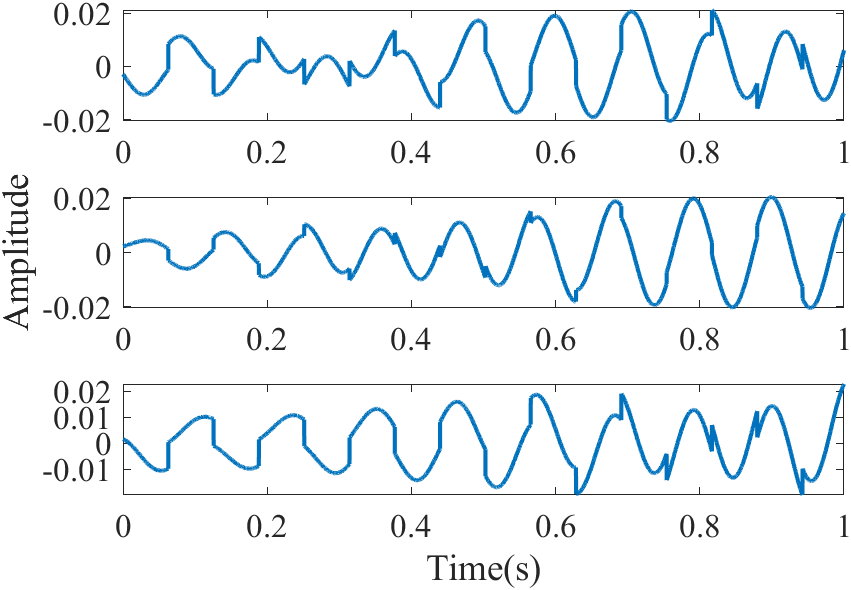}
        \caption{}
        \label{fig:sine_wave_NPSA}
    \end{subfigure}
    \begin{subfigure}{.3\textwidth}
        \centering
        \includegraphics[width=\textwidth]{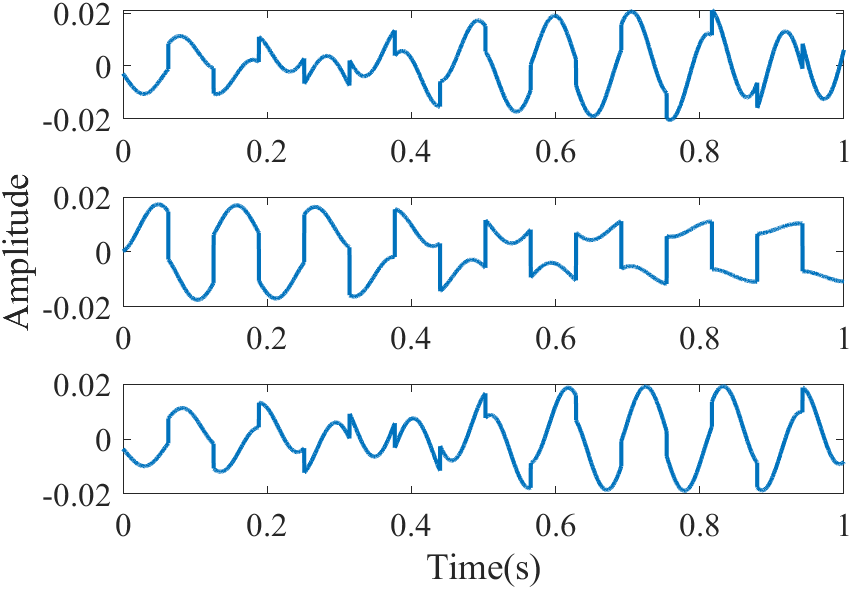}
        \caption{}
        \label{fig:sine_wave_RPSA}
    \end{subfigure}
    \begin{subfigure}{.3\textwidth}
        \centering
        \includegraphics[width=\textwidth]{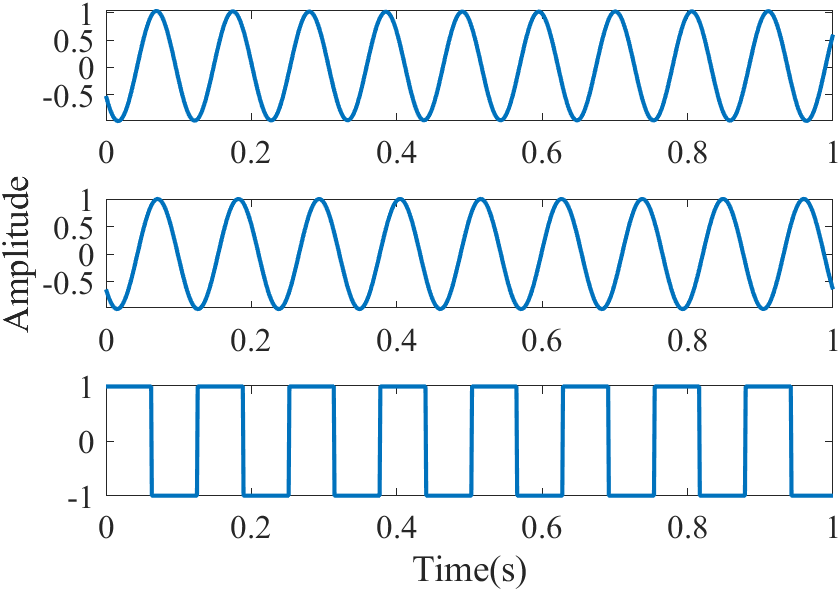}
        \caption{}
        \label{fig:sine_wave_PKA}
    \end{subfigure}
    \caption{The results of the basic waves separation experiment (only the real part of the signals is plotted). (a) source signals (b) mixed signals (c) separation results of JADE (d) separation results of cFastICA (e) separation results of PSA (f) separation results of NPSA (g) separation results of RPSA (h) separation results of PKA}
    \label{fig:simulationExp}
\end{figure}

The separation results of all algorithms can be seen in \cref{fig:simulationExp}. From the results, it can be found that only PKA can separate the signals correctly. For JADE, cFastICA, and PSA, they suppose the source signals are orthogonal to each other, hence, an orthogonal constraint is adopted, which makes them unable to separate nonorthogonal signals. Furthermore, PSA, NPSA, and RPSA use skewness as the non-Gaussianity metric, which naturally hinders their ability to separate symmetrically distributed signals, since the skewness of any symmetric distributed signals is zero. In addition to the visual inspection of the separation results, we also compute some evaluation metrics, including the Intersymbol-interference (ISI) \cite{moreau1994one}, Average Correlation Coefficient (ACC), and the Signal-to-Distortion Ratio (SDR) \cite{vincent2007first, sawada2013multichannel} first to quantify the performance of the algorithms, and the results are shown in \cref{tab:simulationTab}. From \cref{tab:simulationTab}, it can be found that PKA has the best performance in all metrics, which further confirms the superiority of PKA in nonorthogonal signal separation tasks.

\begin{table}[htbp!]
    \centering
    \caption{Basic wave separation performance of all algorithms.}
    \label{tab:simulationTab}
    \begin{tabular}{c|c|c|c|c|c|c}
        \hline
        \text{Index} & \text{JADE} & \text{cFastICA} & \text{PSA} & \text{NPSA} & \text{RPSA} & \text{PKA}  \\ \hline \hline
        \text{ISI}   & 3.9756      & 1.1917          & 3.4394     & 2.8562      & 2.8164      & \bf{0.0316} \\
        \text{ACC}   & 0.8054      & 0.9181          & 0.6372     & 0.5692      & 0.5382      & \bf{0.9980} \\
        \text{SDR}   & 8.065       & 23.99           & -1.464     & -3.14       & -3.901      & \bf{44.16}  \\ \hline
    \end{tabular}
\end{table}

\subsection{Experiment on Sound Data}

In this experiment, we delve into the application of the discussed algorithms on real sound data. Specifically, we select four segments (as can be seen in \cref{fig:soundImg}.) from the ``LibriSpeech ASR corpus'' dataset \cite{panayotov2015librispeech}, a comprehensive collection of approximately 1000 hours of 16kHz English speech recordings. These segments are then combined using a random mixing matrix $\mathbf{A}$, setting the stage for our signal separation experiment.

\begin{figure}[htb!]
    \centering
    \begin{subfigure}{.23\textwidth}
        \centering
        \includegraphics[width=\textwidth]{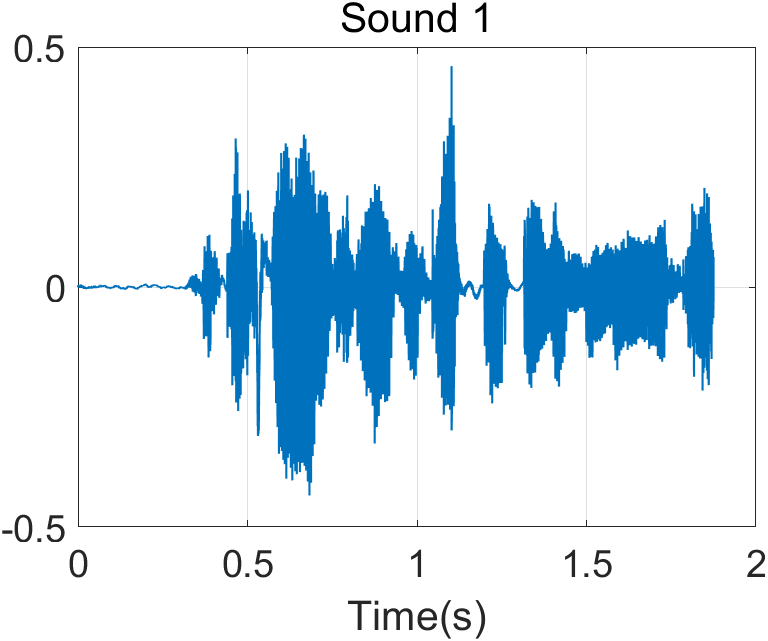}
        \caption{}
        \label{fig:soundImg1}
    \end{subfigure}
    \begin{subfigure}{.23\textwidth}
        \centering
        \includegraphics[width=\textwidth]{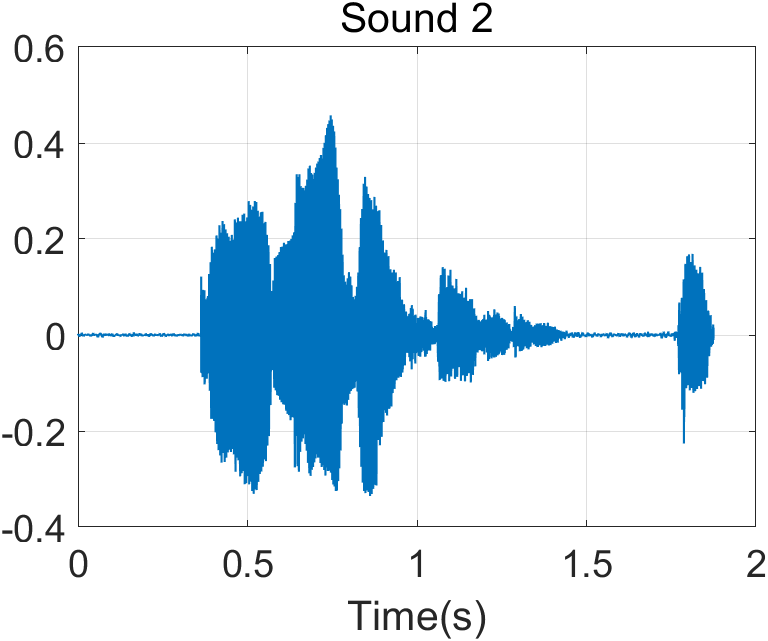}
        \caption{}
        \label{fig:soundImg2}
    \end{subfigure}
    \begin{subfigure}{.23\textwidth}
        \centering
        \includegraphics[width=\textwidth]{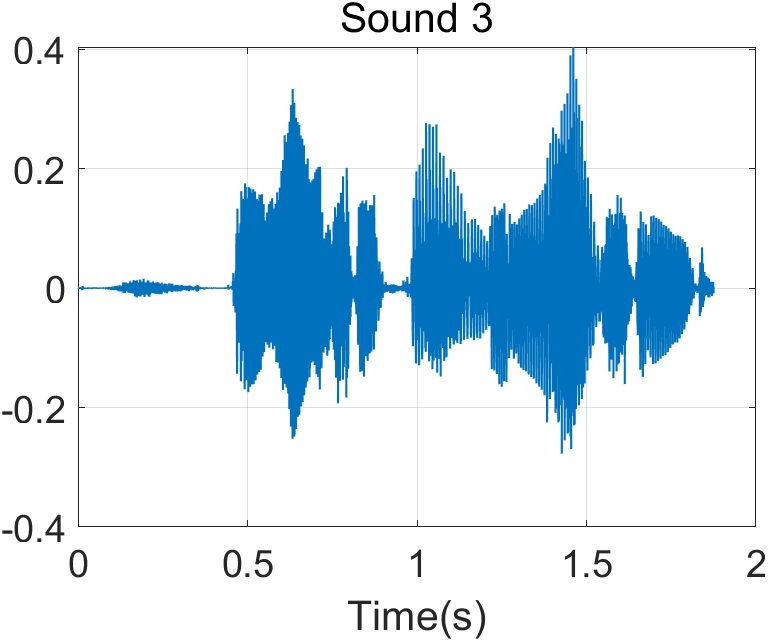}
        \caption{}
        \label{fig:soundImg3}
    \end{subfigure}
    \begin{subfigure}{.23\textwidth}
        \centering
        \includegraphics[width=\textwidth]{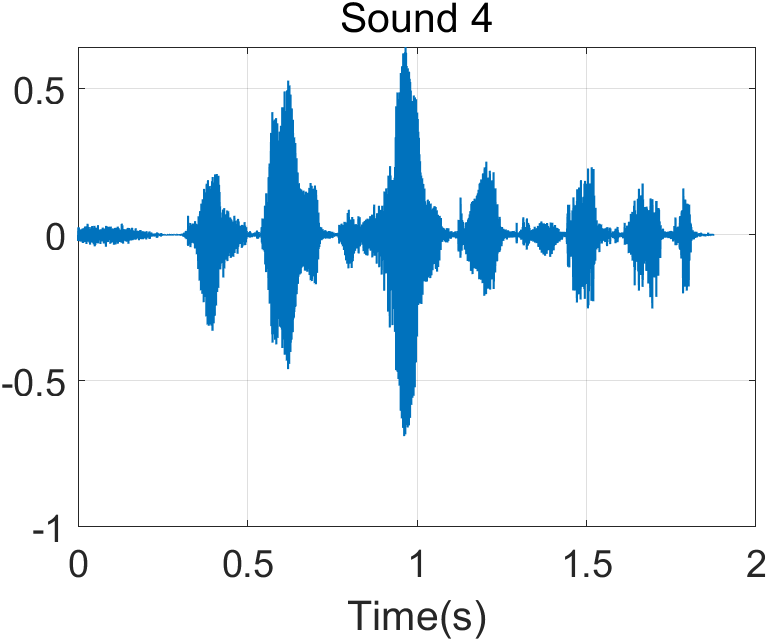}
        \caption{}
        \label{fig:soundImg4}
    \end{subfigure}
    \caption{The waveform of the four selected sound data.}
    \label{fig:soundImg}
\end{figure}

The experimental setup and parameter configurations mirror those employed in the basic wave experiments, with one notable exception: the optimization approach for PKA is adjusted to the gradient ascent method. This modification is made to adapt to the specific characteristics of sound data processing. The results of this experiment are presented in \cref{tab:audioTab}, from which we can find that the skewness-based algorithms, i.e., PSA, NPSA, and RPSA, have a relatively better performance than that in the basic wave experiment. However, the performance of PKA is still the best among all algorithms. The reason is that kurtosis is a more suitable non-Gaussianity metric for sound data separation tasks, and the proposed PKA algorithm can theoretically achieve the optimal solution.

\begin{table}[htbp!]
    \centering
    \caption{Sound data separation performance of JADE, cFastICA, PSA, NPSA, RPSA, and PKA. An average result of 100 runs is computed}
    \label{tab:audioTab}
    \begin{tabular}{c|c|c|c|c|c|c}
        \hline
        \text{Index} & \text{JADE} & \text{cFastICA} & \text{PSA} & \text{NPSA} & \text{RPSA} & \text{PKA}   \\ \hline \hline
        \text{ISI}   & 0.0448      & 0.1110          & 0.0362     & 0.0603      & 0.2897      & \bf{0.0295}  \\
        \text{ACC}   & 0.9974      & 0.9935          & 0.9977     & 0.9967      & 0.9886      & \bf{0.9978}  \\
        \text{SDR}   & 25.2259     & 21.8534         & 24.3058    & 22.6385     & 21.9083     & \bf{28.6976} \\ \hline
    \end{tabular}
\end{table}

\subsection{Experiment on Radar Data}
In this part, we assess the effectiveness of various methods in the task of suppressing radar interference. In a linear array radar system, the signals received are essentially a blend of signals originating from various angles. For this experiment, we have chosen the Linear Frequency Modulation (LFM) signal \cite{cowell2010separation}, a widely utilized signal in radar applications, as the send signal. The types of interference selected for this study include Comb Spectrum Interference (CSI) \cite{ma2022jamming} and Interrupted-Sampling Repeater Jamming (ISRJ) \cite{wang2007mathematic}. It should be noted that CSI is generally orthogonal to the target signal. Meanwhile, ISRJ can be considered as a repeat of the send signal of the radar, hence, it could be nonorthogonal to the desired target signal.

Assume that the incoming angles of the target signal and the interference are $\theta_1$ and $\theta_2$, respectively, and the beam width of the array is $\theta_{\text{3dB}}$. For both CSI and ISRJ, we manually vary the normalized angle difference $\Delta \theta = \frac{|\theta_1 - \theta_2|}{\theta_{\text{3dB}}}$ between the target signal and the interference. Additionally, we also adjust the Signal-to-Noise Ratio (SNR) and the Signal-to-Interference Ratio (SIR) of the incoming data.

All algorithms are executed with the aim of separating the target signal and suppressing the interference. The effectiveness of the interference suppression is quantified by the improvement in the SIR after processing. The results of these evaluations are presented in \cref{fig:radarExp2} and \cref{fig:radarExp}, providing a comparative analysis of the performance of the algorithms under varying conditions of angle difference, SNR, and SIR.

\begin{figure}[htb!]
    \centering
    \begin{subfigure}{.32\textwidth}
        \centering
        \includegraphics[width=.9\textwidth]{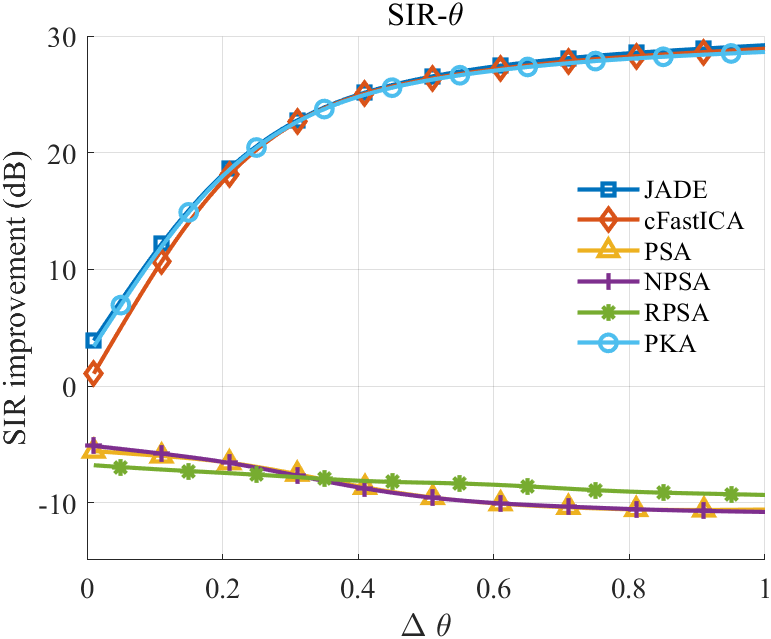}
        \caption{}
    \end{subfigure}
    \begin{subfigure}{.32\textwidth}
        \centering
        \includegraphics[width=.9\textwidth]{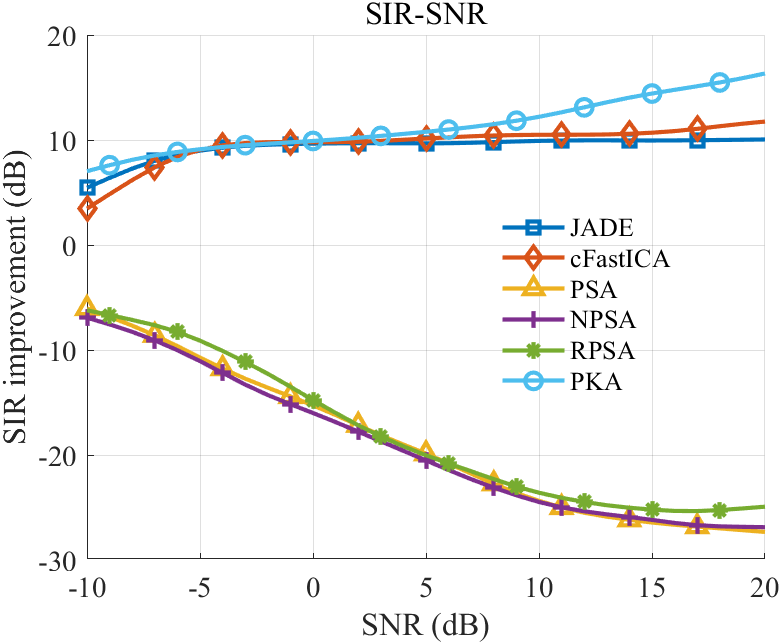}
        \caption{}
        \label{fig:radarExp2b}
    \end{subfigure}
    \begin{subfigure}{.32\textwidth}
        \centering
        \includegraphics[width=.9\textwidth]{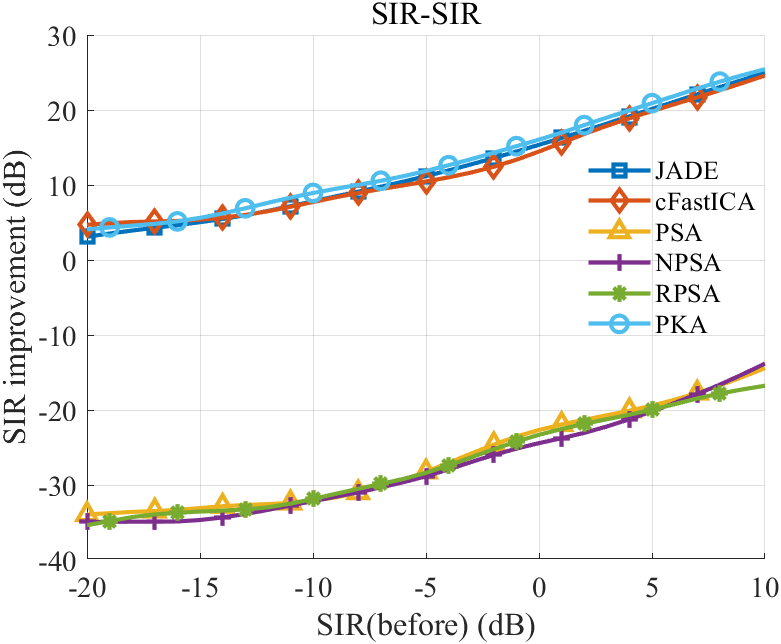}
        \caption{}
        \label{fig:radarExp2c}
    \end{subfigure}
    \caption{The results of CSI suppressing experiment. (a) the variation of SIR improvement with angle difference $\Delta \theta$, where SNR=$10\text{dB}$, SIR=$0\text{dB}$, (b) the variation of SIR improvement with SNR, where $\Delta \theta=1$, SIR=$0\text{dB}$, (c) the variation of SIR improvement after suppression with SIR before suppression, where SNR=$10\text{dB}$, $\Delta \theta=1$}
    \label{fig:radarExp2}
\end{figure}

\begin{figure}[htb!]
    \centering
    \begin{subfigure}{.32\textwidth}
        \centering
        \includegraphics[width=.9\textwidth]{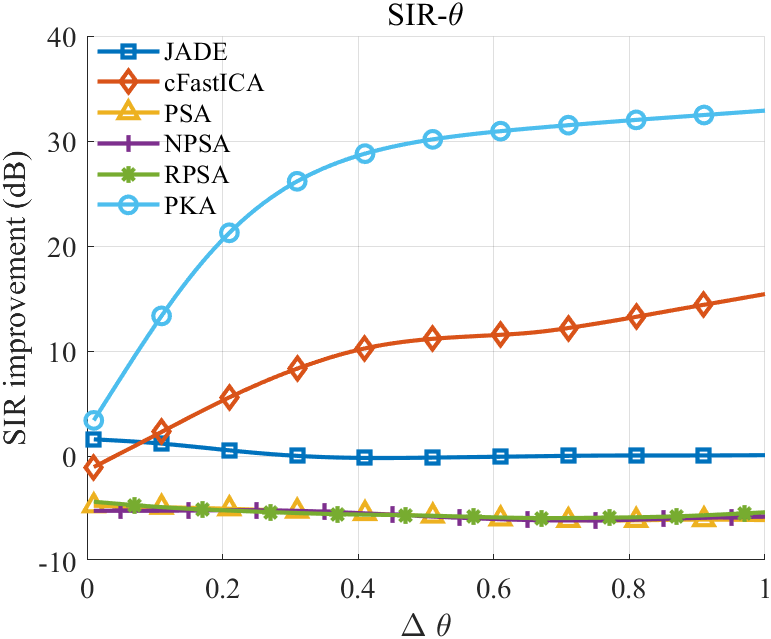}
        \caption{}
    \end{subfigure}
    \begin{subfigure}{.32\textwidth}
        \centering
        \includegraphics[width=.9\textwidth]{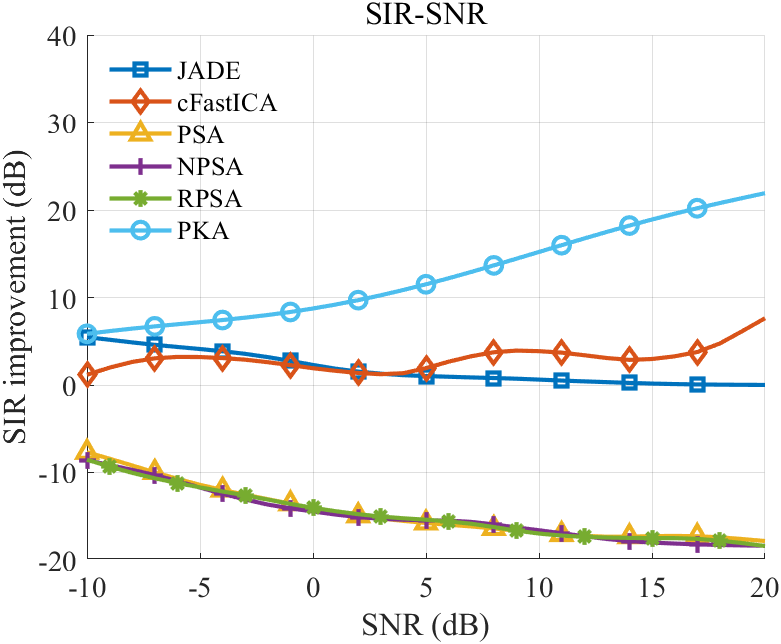}
        \caption{}
    \end{subfigure}
    \begin{subfigure}{.32\textwidth}
        \centering
        \includegraphics[width=.9\textwidth]{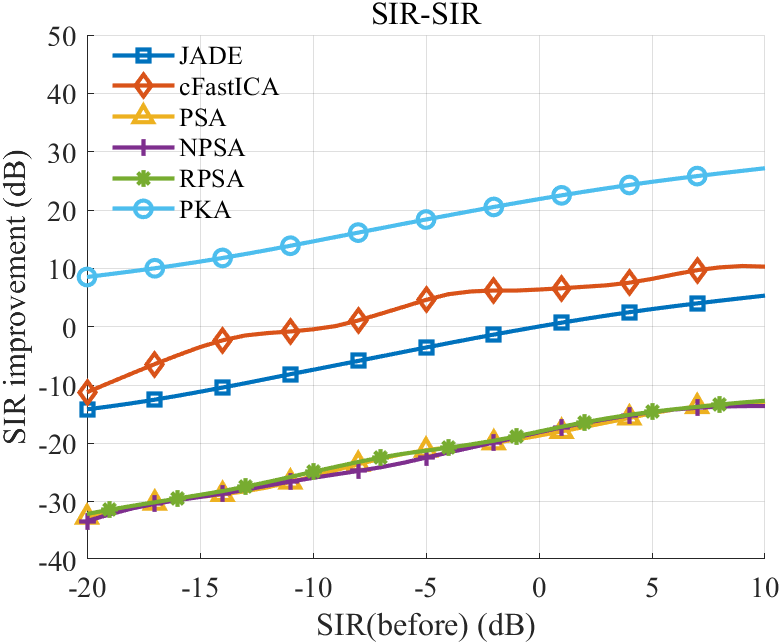}
        \caption{}
    \end{subfigure}
    \caption{The results of ISRJ suppressing experiment. (a) the variation of SIR improvement with the angle difference $\Delta \theta$, where SNR=$10\text{dB}$, SIR=$0\text{dB}$, (b) the variation of SIR improvement with SNR, where $\Delta \theta=1$, SIR=$0\text{dB}$, (c) the variation of SIR improvement after suppression with SIR before suppression, where SNR=$10\text{dB}$, $\Delta \theta=1$}
    \label{fig:radarExp}
\end{figure}

It can be found in \cref{fig:radarExp2,fig:radarExp} that for the complex radar signals, the skewness-based algorithms, i.e., PSA, NPSA, and RPSA are unable to suppress the interference effectively, since both the signal and the interference having a skewness value of zero. In contrast, the kurtosis-based algorithms, i.e., JADE, cFastICA, and PKA, have a better performance in the radar interference suppression task. From \cref{fig:radarExp2}, it can be seen that as the angle difference between the CSI interference and the target, SNR, and SIR increase, the performance of all kurtosis-based algorithms improves. Moreover, since the interference is orthogonal to the target signal, PKA has a similar performance to JADE and cFastICA, and all of them outperform the other algorithms. As for the ISRJ interference (\cref{fig:radarExp}), due to the interference and target signal being nonorthogonal, which contradicts the assumptions of JADE and cFastICA about the source signals, their performance inevitably deteriorates. Meanwhile, from \cref{fig:radarExp}, it can be found that the performance of PKA is always better than the other algorithms, which further confirms the superiority of PKA in nonorthogonal signal separation tasks.

%% file: tex/6-conclusion.tex
\section{Conclusion}
Orthogonal constraints are commonly employed in ICA algorithms to separate independent sources. However, these constraints imply that the sources are mutually independent (orthogonal to each other), which is not always the case in practical scenarios. Consequently, the performance of orthogonal constraint-based algorithms is limited when dealing with non-orthogonal sources. In this paper, we introduce a non-zero volume constraint and a Riemannian gradient-based algorithm to accurately and effectively separate non-orthogonal sources. Additionally, we extend tensor-based algorithms to the complex domain by utilizing the fourth-order statistic (kurtosis) as the measure of non-Gaussianity. Experiments on both synthetic and real-world datasets demonstrate the effectiveness of the proposed algorithm. However, during the experiments, we observed that the proposed algorithm is sensitive to the initialization of the transformation matrix. In future work, we plan to investigate more robust initialization methods to enhance the performance of the proposed algorithm.

%% file: arxiv/ref.bib
@inproceedings{huang2004radar,
    title        = {A radar anti-jamming technology based on blind source separation},
    author       = {Huang, Gaoming and Yang, Luxi},
    booktitle    = {Proceedings 7th International Conference on Signal Processing, 2004. Proceedings. ICSP'04. 2004.},
    volume       = {3},
    pages        = {2021--2024},
    year         = {2004},
    organization = {IEEE}
}

@article{hyvarinen1997fast,
    title     = {A fast fixed-point algorithm for independent component analysis},
    author    = {Hyv{\"a}rinen, Aapo and Oja, Erkki},
    journal   = {Neural computation},
    volume    = {9},
    number    = {7},
    pages     = {1483--1492},
    year      = {1997},
    publisher = {MIT Press One Rogers Street, Cambridge, MA 02142-1209, USA journals-info~…}
}

@inproceedings{cardoso1993blind,
    title        = {Blind beamforming for non-Gaussian signals},
    author       = {Cardoso, Jean-Fran{\c{c}}ois and Souloumiac, Antoine},
    booktitle    = {IEE proceedings F (radar and signal processing)},
    volume       = {140},
    number       = {6},
    pages        = {362--370},
    year         = {1993},
    organization = {IET}
}

@inproceedings{kim2006independent,
    title        = {Independent vector analysis: Definition and algorithms},
    author       = {Kim, Taesu and Lee, Intae and Lee, Te-Won},
    booktitle    = {2006 Fortieth Asilomar Conference on Signals, Systems and Computers},
    pages        = {1393--1396},
    year         = {2006},
    organization = {IEEE}
}

@inproceedings{kurita2000evaluation,
    title        = {Evaluation of blind signal separation method using directivity pattern under reverberant conditions},
    author       = {Kurita, Satoshi and Saruwatari, Hiroshi and Kajita, Shoji and Takeda, Kazuya and Itakura, Fumitada},
    booktitle    = {2000 IEEE International Conference on Acoustics, Speech, and Signal Processing. Proceedings (Cat. No. 00CH37100)},
    volume       = {5},
    pages        = {3140--3143},
    year         = {2000},
    organization = {IEEE}
}

@article{yatabe2021determined,
    title     = {Determined BSS based on time-frequency masking and its application to harmonic vector analysis},
    author    = {Yatabe, Kohei and Kitamura, Daichi},
    journal   = {IEEE/ACM Transactions on Audio, Speech, and Language Processing},
    volume    = {29},
    pages     = {1609--1625},
    year      = {2021},
    publisher = {IEEE}
}

@article{vincent2006performance,
  title={Performance measurement in blind audio source separation},
  author={Vincent, Emmanuel and Gribonval, R{\'e}mi and F{\'e}votte, C{\'e}dric},
  journal={IEEE transactions on audio, speech, and language processing},
  volume={14},
  number={4},
  pages={1462--1469},
  year={2006},
  publisher={IEEE}
}

@article{geng2014principal,
  title={Principal skewness analysis: Algorithm and its application for multispectral/hyperspectral images indexing},
  author={Geng, Xiurui and Ji, Luyan and Sun, Kang},
  journal={IEEE Geoscience and Remote Sensing Letters},
  volume={11},
  number={10},
  pages={1821--1825},
  year={2014},
  publisher={IEEE}
}

@article{geng2020npsa,
  title={NPSA: Nonorthogonal principal skewness analysis},
  author={Geng, Xiurui and Wang, Lei},
  journal={IEEE Transactions on Image Processing},
  volume={29},
  pages={6396--6408},
  year={2020},
  publisher={IEEE}
}

@inproceedings{murata1998line,
  title={An on-line algorithm for blind source separation on speech signals},
  author={Murata, Noboru and Ikeda, Shiro},
  booktitle={Proc. NOLTA98},
  volume={3},
  pages={923--926},
  year={1998},
  organization={Citeseer}
}

@article{jutten1988solution,
  title={Une solution neuromim{\'e}tique au probl{\`e}me de s{\'e}paration de sources},
  author={Jutten, Christian and H{\'e}rault, J},
  journal={Traitement du signal},
  volume={5},
  number={6},
  pages={389--403},
  year={1988}
}

@inproceedings{hiroe2006solution,
  title={Solution of permutation problem in frequency domain ICA, using multivariate probability density functions},
  author={Hiroe, Atsuo},
  booktitle={Independent Component Analysis and Blind Signal Separation: 6th International Conference, ICA 2006, Charleston, SC, USA, March 5-8, 2006. Proceedings 6},
  pages={601--608},
  year={2006},
  organization={Springer}
}

@article{makino2005blind,
  title={Blind source separation of convolutive mixtures of speech in frequency domain},
  author={Makino, Shoji and Sawada, Hiroshi and Mukai, Ryo and Araki, Shoko},
  journal={IEICE transactions on fundamentals of electronics, communications and computer sciences},
  volume={88},
  number={7},
  pages={1640--1655},
  year={2005},
  publisher={The Institute of Electronics, Information and Communication Engineers}
}

@inproceedings{liutkus2013overview,
  title={An overview of informed audio source separation},
  author={Liutkus, Antoine and Durrieu, Jean-Louis and Daudet, Laurent and Richard, Ga{\"e}l},
  booktitle={2013 14th International Workshop on Image Analysis for Multimedia Interactive Services (WIAMIS)},
  pages={1--4},
  year={2013},
  organization={IEEE}
}

@inproceedings{makino2006blind,
  title={Blind source separation of convolutive mixtures},
  author={Makino, Shoji},
  booktitle={Independent Component Analyses, Wavelets, Unsupervised Smart Sensors, and Neural Networks IV},
  volume={6247},
  pages={58--72},
  year={2006},
  organization={SPIE}
}

@article{fabrizio2014blind,
  title={Blind source separation with the generalised estimation of multipath signals algorithm},
  author={Fabrizio, Giuseppe and Farina, Alfonso},
  journal={IET Radar, Sonar \& Navigation},
  volume={8},
  number={9},
  pages={1255--1266},
  year={2014},
  publisher={Wiley Online Library}
}

@article{guo2017time,
  title={A time-frequency domain underdetermined blind source separation algorithm for MIMO radar signals},
  author={Guo, Qiang and Ruan, Guoqing and Liao, Yanping},
  journal={Symmetry},
  volume={9},
  number={7},
  pages={104},
  year={2017},
  publisher={MDPI}
}

@article{meganem2014linear,
  title={Linear-quadratic blind source separation using NMF to unmix urban hyperspectral images},
  author={Meganem, Ines and Deville, Yannick and Hosseini, Shahram and Deliot, Philippe and Briottet, Xavier},
  journal={IEEE Transactions on Signal Processing},
  volume={62},
  number={7},
  pages={1822--1833},
  year={2014},
  publisher={IEEE}
}

@inproceedings{golbabaee2010distributed,
  title={Distributed compressed sensing of hyperspectral images via blind source separation},
  author={Golbabaee, Mohammad and Arberet, Simon and Vandergheynst, Pierre},
  booktitle={2010 Conference Record of the Forty Fourth Asilomar Conference on Signals, Systems and Computers},
  pages={196--198},
  year={2010},
  organization={IEEE}
}

@article{Bingham2000A,
  title={A Fast Fixed-Point Algorithm for Independent Component Analysis of Complex Valued Signals},
  author={ Bingham, E. },
  journal={International Journal of Neural Systems},
  volume={10},
  number={1},
  pages={1-8},
  year={2000},
}

@article{kemiha2017complex,
  title={Complex blind source separation},
  author={Kemiha, Mina and Kacha, Abdellah},
  journal={Circuits, Systems, and Signal Processing},
  volume={36},
  pages={4670--4687},
  year={2017},
  publisher={Springer}
}

@article{bell1995information,
  title={An information-maximization approach to blind separation and blind deconvolution},
  author={Bell, Anthony J and Sejnowski, Terrence J},
  journal={Neural computation},
  volume={7},
  number={6},
  pages={1129--1159},
  year={1995},
  publisher={MIT Press}
}

@article{oja1995signal,
  title={Signal separation by nonlinear Hebbian learning},
  author={Oja, Erkki and Karhunen, Juha},
  journal={Computational intelligence: A dynamic system perspective},
  pages={83--97},
  year={1995},
  publisher={Citeseer}
}

@article{hirayama2015unifying,
  title={Unifying blind separation and clustering for resting-state EEG/MEG functional connectivity analysis},
  author={Hirayama, Jun-ichiro and Ogawa, Takeshi and Hyv{\"a}rinen, Aapo},
  journal={Neural computation},
  volume={27},
  number={7},
  pages={1373--1404},
  year={2015},
  publisher={MIT Press One Rogers Street, Cambridge, MA 02142-1209, USA journals-info~…}
}

@article{choi2005blind,
  title={Blind source separation and independent component analysis: A review},
  author={Choi, Seungjin and Cichocki, Andrzej and Park, Hyung-Min and Lee, Soo-Young},
  journal={Neural Information Processing-Letters and Reviews},
  volume={6},
  number={1},
  pages={1--57},
  year={2005}
}

@inproceedings{moreau1994one,
  title={A one stage self-adaptive algorithm for source separation},
  author={Moreau, Eric and Macchi, Odile},
  booktitle={Proceedings of ICASSP'94. IEEE International Conference on Acoustics, Speech and Signal Processing},
  volume={3},
  pages={III--49},
  year={1994},
  organization={IEEE}
}

@inproceedings{vincent2007first,
  title={First stereo audio source separation evaluation campaign: data, algorithms and results},
  author={Vincent, Emmanuel and Sawada, Hiroshi and Bofill, Pau and Makino, Shoji and Rosca, Justinian P},
  booktitle={International Conference on Independent Component Analysis and Signal Separation},
  pages={552--559},
  year={2007},
  organization={Springer}
}

@article{sawada2013multichannel,
  title={Multichannel extensions of non-negative matrix factorization with complex-valued data},
  author={Sawada, Hiroshi and Kameoka, Hirokazu and Araki, Shoko and Ueda, Naonori},
  journal={IEEE Transactions on Audio, Speech, and Language Processing},
  volume={21},
  number={5},
  pages={971--982},
  year={2013},
  publisher={IEEE}
}

@article{wang2007mathematic,
  title={Mathematic principles of interrupted-sampling repeater jamming (ISRJ)},
  author={Wang, XueSong and Liu, JianCheng and Zhang, WenMing and Fu, QiXiang and Liu, Zhong and Xie, XiaoXia},
  journal={Science in China Series F: Information Sciences},
  volume={50},
  pages={113--123},
  year={2007},
  publisher={Springer}
}

@mastersthesis{ma2022jamming,
  author  = {Minghui Ma},
  school  = {Hefei University of Technology},
  title   = {Research on Radar Active Jamming Recognition Algorithm},
  year    = {2022}
}

@article{kitamura2016determined,
  title={Determined blind source separation unifying independent vector analysis and nonnegative matrix factorization},
  author={Kitamura, Daichi and Ono, Nobutaka and Sawada, Hiroshi and Kameoka, Hirokazu and Saruwatari, Hiroshi},
  journal={IEEE/ACM Transactions on Audio, Speech, and Language Processing},
  volume={24},
  number={9},
  pages={1626--1641},
  year={2016},
  publisher={IEEE}
}

@article{kitamura2018determined,
  title={Determined blind source separation with independent low-rank matrix analysis},
  author={Kitamura, Daichi and Ono, Nobutaka and Sawada, Hiroshi and Kameoka, Hirokazu and Saruwatari, Hiroshi},
  journal={Audio source separation},
  pages={125--155},
  year={2018},
  publisher={Springer}
}

@article{kitamura2020consistent,
  title={Consistent independent low-rank matrix analysis for determined blind source separation},
  author={Kitamura, Daichi and Yatabe, Kohei},
  journal={EURASIP journal on advances in signal processing},
  volume={2020},
  number={1},
  pages={1--35},
  year={2020},
  publisher={SpringerOpen}
}

@ARTICLE{sawada2023multi,
  author={Sawada, Hiroshi and Ikeshita, Rintaro and Kinoshita, Keisuke and Nakatani, Tomohiro},
  journal={IEEE/ACM Transactions on Audio, Speech, and Language Processing}, 
  title={Multi-Frame Full-Rank Spatial Covariance Analysis for Underdetermined Blind Source Separation and Dereverberation}, 
  year={2023},
  volume={31},
  number={},
  pages={3589-3602},
  doi={10.1109/TASLP.2023.3313446}
}

@article{bonnabel2013stochastic,
  title={Stochastic gradient descent on Riemannian manifolds},
  author={Bonnabel, Silvere},
  journal={IEEE Transactions on Automatic Control},
  volume={58},
  number={9},
  pages={2217--2229},
  year={2013},
  publisher={IEEE}
}

@inproceedings{panayotov2015librispeech,
  title={Librispeech: an asr corpus based on public domain audio books},
  author={Panayotov, Vassil and Chen, Guoguo and Povey, Daniel and Khudanpur, Sanjeev},
  booktitle={2015 IEEE international conference on acoustics, speech and signal processing (ICASSP)},
  pages={5206--5210},
  year={2015},
  organization={IEEE}
}

@article{cowell2010separation,
  title={Separation of overlapping linear frequency modulated (LFM) signals using the fractional Fourier transform},
  author={Cowell, David MJ and Freear, Steven},
  journal={IEEE transactions on ultrasonics, ferroelectrics, and frequency control},
  volume={57},
  number={10},
  pages={2324--2333},
  year={2010},
  publisher={IEEE}
}

@article{geng4438906rpsa,
  title={Rpsa: Riemannian Principal Skewness Analysis},
  author={Geng, Xiurui and Gao, Jingyu},
  journal={Available at SSRN 4438906}
}
